\documentclass[12pt,a4paper]{article}
\usepackage{amsmath,amsthm,amssymb,amscd}
\usepackage{latexsym}
\usepackage{indentfirst}
\usepackage{bibentry}
\usepackage{textcomp}
\usepackage{float}
\usepackage{multirow}
\usepackage{booktabs}
\usepackage{graphicx}
\usepackage{color}
\usepackage[nosort]{cite}
\usepackage{authblk}
\usepackage[numbers]{natbib}
\setcitestyle{open={},close={}}

\makeatletter \@addtoreset{equation}{section}

\makeatother
\newtheorem{thm}{Theorem}[section]
\newtheorem{cor}{Corollary}[section]
\newtheorem{lem}{Lemma}[section]

\newtheorem{pro}{Proposition}[section]

\theoremstyle{definition}
\newtheorem{rem}{Remark}[section]

\begin{document}

\title{On properties of a deformed Freud weight}
\author[1]{{Mengkun Zhu}\footnote{Zhu\_mengkun@163.com; zhu.mengkun@connect.umac.mo}}
\author[1]{{Yang Chen}\footnote{yangbrookchen@yahoo.co.uk}}
\affil[1]{Department of Mathematics, University of Macau,
Avenida da Universidade, Taipa, Macau, China}

\renewcommand\Authands{ and }

\maketitle

\begin{abstract}
We study the recurrence coefficients of the monic polynomials $P_n(z)$ orthogonal with respect to the deformed (also called semi-classical) Freud weight
\begin{equation*}
w_{\alpha}(x;s,N)=|x|^{\alpha}{\rm e}^{-N\left[x^{2}+s\left(x^{4}-x^{2}\right)\right]},
~~x\in\mathbb{R},
\end{equation*}
with parameters $\alpha>-1,~N>0,~s\in[0,1]$. We show that the recurrence coefficients $\beta_{n}(s)$ satisfy the first discrete Painlev\'{e} equation (denoted by d${\rm P_{I}}$), a differential-difference equation and a second order nolinear ordinary differential equation (ODE) in $s$. Here $n$ is the order of the Hankel matrix generated by $w_{\alpha}(x;s,N)$. We describe the asymptotic behavior of the  recurrence coefficients in three situations, (i) $s\rightarrow0$, $n,N$ finite, (ii) $n\rightarrow\infty$, $N$ finite, (iii) $n, N\rightarrow\infty$, such that the radio $r:=\frac{n}{N}$ is bounded away from $0$ and closed to $1$. We also investigate the existence and uniqueness for the positive solutions of the d${\rm P_{I}}$.

Further more, we derive, using the ladder approach, a second order linear ODE satisfied by the polynomials $P_n(z)$. It is found as $n\rightarrow\infty$, the linear ODE turns to be a biconfluent Heun equation. This paper concludes with the study of the Hankel determinant, $D_{n}(s)$, associated with $w_{\alpha}(x;s,N)$ when $n$ tends to infinity.

\end{abstract}

\section{Introduction}
The orthogonal polynomials, recurrence coefficients, eigenvalues and the relevant properties associated with the exponential weights ${\rm e}^{-|x|^{m}}$, $m\in\mathbb{N}$ as well as various generalized forms of this weight function on $\mathbb{R}$, have been investigated by many authors, Refs. [\cite{Its,ChenFeigin,ChenLubinsky,PAC2,PAC3,C2,Filipuk,Magnus,C3}, \cite{C6}$-$\cite{Noschese}].

Let $P_{n}(z)$ be the monic polynomials of degree $n$ orthogonal with respect to the deformed Freud weight
\begin{equation}\label{weight}
w(x)=w_{\alpha}(x;s,N)=|x|^{\alpha}{\rm e}^{-N\left[x^{2}+s\left(x^{4}-x^{2}\right)\right]}, ~x\in\mathbb{R},~\alpha>-1,~N>0,~s\in[0,1],
\end{equation}
where the deformation $s\left(x^{4}-x^{2}\right)$, given by $0\leq s\leq1$, interpolates between the generalized Gaussian (Hermite) weight ($|x|^{\alpha}{\rm e}^{-Nx^{2}},~x\in\mathbb{R},~\alpha>-1,~N>0$) when $s=0$ and the Freud weight ($|x|^{\alpha}{\rm e}^{-Nx^{4}},~x\in\mathbb{R},~\alpha>-1,~N>0$) when $s=1$.

From the orthogonality condition, given by
\begin{equation}\label{intro1}
\int_{\mathbb{R}}P_{j}(x)P_{k}(x)w_{\alpha}(x;s,N)dx=h_{j}(s;\alpha,N)\delta_{jk},~~h_{j}(s;\alpha,N)>0,~~j,k\in\{0,1,2,\ldots\},
\end{equation}
where $\delta_{jk}$ denotes the Kronecker delta, and $h_{j}$ is the square of $L^{2}$ norm of the monic polynomial $P_{j}(x)$, there follows the three term recurrence relation,
\begin{equation}\label{intro2}
zP_{n}(z;s,\alpha,N)=P_{n+1}(z;s,\alpha,N)+\beta_{n}(s;\alpha,N)P_{n-1}(z;s,\alpha,N),~~n\geq0,
\end{equation}
subject to the initial conditions
\begin{equation*}
P_{0}(z):=1~~~{\rm and}~~~\beta_{0}P_{-1}(z):=0.
\end{equation*}
Multiplying both sides of Eq. (\ref{intro2}) by $P_{n-1}(z)w_{\alpha}(z;s,N)$ and integrating this with respect to $z$ on $\mathbb{R}$, which, due to the orthogonality condition (\ref{intro1}), gives us
\begin{equation*}
\beta_{n}(s)=\frac{1}{h_{n-1}(s)}\int_{\mathbb{R}}zP_{n}(z)P_{n-1}(z)w_{\alpha}(z;s,N)dx=\frac{h_{n}(s)}{h_{n-1}(s)}>0.
\end{equation*}
It should be pointed out that $P_{n}(z)$ contains only the terms $z^{n-j}$, $j\leq n$ and even, since our weight function $w_{\alpha}(x;s,N)$ is even on $\mathbb{R}$. This implies that
\begin{equation*}
P_{n}(-z)=(-1)^{n}P_{n}(z) ~~~{\rm and}~~~P_{n}(0)P_{n-1}(0)=0.
\end{equation*}
Then we note the monic polynomials $P_{n}(z)$, associated with $w_{\alpha}(x;s,N)$, can be normalized as [\cite{bbbb}],
\begin{equation*}
P_{2j}(z)=z^{2j}+\textbf{p}(2j;s)z^{2j-2}+\cdots+P_{2j}(0),
\end{equation*}
and
\begin{equation*}
\begin{split}
P_{2j+1}(z)&=z^{2j+1}+\textbf{p}(2j+1;s)z^{2j-1}+\cdots+{\rm const.}z\\
&=z\left(z^{2j}+\textbf{p}(2j+1;s)z^{2j-2}+\cdots+{\rm const.}\right).
\end{split}
\end{equation*}

In unitary random matrix theory with the Hermitian ensemble, the joint probability density function of the $n$ eigenvalues $\left\{x_{j}\right\}_{j=1}^{n}$ is given in [\cite{Mehta}] by
\begin{equation*}
p\left(x_{1},\ldots,x_{n}\right)\prod_{j=1}^{n}dx_{j}=\frac{1}{D_{n}(s)}\frac{1}{n!}\prod_{1\leq j<k\leq n}\left(x_{j}-x_{k}\right)^{2}\prod_{\ell=1}^{n}w(x_{\ell})dx_{\ell},
\end{equation*}
where $D_{n}(s)$ denotes the normalization constant (also called partition function), which reads
\begin{equation*}
D_{n}(s)=\frac{1}{n!}\int_{\mathbb{R}^{n}}\prod_{1\leq j<k\leq n}\left(x_{k}-x_{j}\right)^{2}\prod_{\ell=1}^{n}w(x_{\ell})dx_{\ell},
\end{equation*}
so that
\begin{equation*}
\int_{\mathbb{R}^{n}}p\left(x_{1},\ldots,x_{n}\right)\prod_{j=1}^{n}dx_{j}=1.
\end{equation*}
For the problem at hand, the weight function $w(x)$ is given by
\begin{equation*}
w(x)=|x|^{\alpha}{\rm e}^{-N\left[x^{2}+s\left(x^{4}-x^{2}\right)\right]}, ~x\in\mathbb{R},~\alpha>-1,~N>0,~s\in[0,1],
\end{equation*}
and the moments are defined by
\begin{equation*}
\mu_{j}(s;\alpha,N)=\int_{\mathbb{R}}x^{j}|x|^{\alpha}{\rm e}^{-N\left[x^{2}+s\left(x^{4}-x^{2}\right)\right]}dx,~~j\in\{0,1,2,\ldots\}.
\end{equation*}
$D_{n}(s)$ admits two more alternative representations [\cite{szego}],
\begin{equation}\label{Dn}
D_{n}(s):=\det\left(\mu_{j+k}\right)_{j,k=0}^{n-1}=\prod_{j=0}^{n-1}h_{j},
\end{equation}
where $\det\left(\mu_{j+k}\right)_{j,k=0}^{n-1}$ is the determinant of the Hankel matrix (or moment matrix).

For the deformed Freud weight (\ref{weight}), $\mu_{0}(s;\alpha,N)$ can be evaluated in term of the parabolic cylinder (Hermite-Weber) function $D_{\nu}(\zeta)$. By definition
\begin{equation*}
\begin{split}
\mu_{0}(s;\alpha,N):&=\int_{-\infty}^{\infty}|x|^{\alpha}{\rm e}^{-N\left[x^{2}+s\left(x^{4}-x^{2}\right)\right]}dx\\
&=2\int_{0}^{\infty}x^{\alpha}{\rm e}^{-N\left[x^{2}+s\left(x^{4}-x^{2}\right)\right]}dx\\
&=\int_{0}^{\infty}y^{\frac{\alpha-1}{2}}{\rm e}^{-N\left[y+s\left(y^{2}-y\right)\right]}dy\\
&=(2Ns)^{-\frac{\alpha+1}{4}}\Gamma\left(\frac{\alpha+1}{2}\right)\exp\left[\frac{N(1-s)^{2}}{8s}\right]D_{-\frac{\alpha+1}{2}}\left[\frac{N(1-s)}{\sqrt{2Ns}}\right],
\end{split}
\end{equation*}
since the parabolic cylinder function $D_{\nu}(z)$ has the integral representation [\cite{C4}, \S 12.5 (i)]
\begin{equation*}
D_{\nu}(z)=\frac{\exp\left(-\frac{z^{2}}{4}\right)}{\Gamma(-\nu)}\int_{0}^{\infty}t^{-\nu-1}\exp\left(-\frac{t^{2}}{2}-tz \right)dt,~~~\text{$\Re\nu<0$}.
\end{equation*}
We note that the even moments are
\begin{equation*}
\mu_{2n}(s;\alpha,N)=\int_{-\infty}^{\infty}x^{2n}|x|^{\alpha}{\rm e}^{-N\left[x^{2}+s\left(x^{4}-x^{2}\right)\right]}dx=\mu_{0}\left(s;\alpha+2n,N\right),~n\in\mathbb{N},
\end{equation*}
whilst the odd ones are
\begin{equation*}
\mu_{2n+1}(s;\alpha,N)=\int_{-\infty}^{\infty}x^{2n+1}|x|^{\alpha}{\rm e}^{-N\left[x^{2}+s\left(x^{4}-x^{2}\right)\right]}dx=0,~n\in\mathbb{N}.
\end{equation*}

It should be pointed out our recurrence coefficients, moments, $L^{2}$ norm of orthogonal polynomials, Hankel determinant depend on $s$, also parameters $N$ and $\alpha$, but we may not display them unless we have to.
\begin{rem}
The dependence of the orthogonal polynomials $P_{n}(x;s,\alpha,N)$ on $s$, $\alpha$ and $N$ can be seen from its determinant representation in terms of the moments, or alternatively, from the Heine formula [\cite{szego}, Eq. 2.2.10]
\begin{equation*}
\begin{split}
P_{n}(z;s,\alpha,N)&=\frac{1}{D_{n}(s;\alpha,N)}\frac{1}{n!}\int_{\mathbb{R}^{n}}\prod_{1\leq j<k\leq n}\left(x_{j}-x_{k}\right)^{2}\prod_{\ell=1}^{n}\left(z-x_{\ell}\right)w_{\alpha}(x_{\ell};s,N)dx_{\ell}\\
&=\frac{\det\left(\int_{\mathbb{R}}x^{j+k}(z-x)w_{\alpha}(x;s,N)dx\right)_{j,k=0}^{n-1}}{\det\left(\int_{\mathbb{R}}x^{j+k}w_{\alpha}(x;s,N)dx\right)_{j,k=0}^{n-1}}.
\end{split}
\end{equation*}
\end{rem}

The remainder of this paper is organized as follows. In Sect. 2 we prove that the recurrence coefficients satisfy the first discrete Painlev\'{e} equation (${\rm dP_{I}}$), a differential-difference equation and a second order ODE. Sects. 3-5 discuss the asymptotic behavior of the recurrence coefficients in three cases, i.e., $s\rightarrow0$, $n,N$ fixed; $n\rightarrow\infty$, $N$ fixed; $n,N\rightarrow\infty$ but the quantity $n/N$ fixed. In Sect. 6, we obtain a second order ODE satisfied by the deformed Freud orthogonal polynomials applying the ladder operator approach, as well as a biconfluent Heun equation when $n$ tends to infinity. Sect. 7 investigates the existence and uniqueness for the solutions of the ${\rm dP_{I}}$. Finally, in Sect. 8, with the aid of the compatibility conditions for the ladder operators, we find the asymptotic expansion of the logarithm of the Hankel determinant associated with our weight.

\section{Recurrence coefficients of the deformed Freud polynomials}
It is known [\cite{C1,C2,C3}] that the recurrence coefficients satisfy the following second order nonlinear difference equation,
\begin{pro}\label{pro3.1}
The recurrence coefficients $\beta_{n}(s)$ satisfy the first discrete Painlev\'{e} equation, ${\rm dP_{I}}$,
\begin{equation}\label{eq3}
\beta_{n+1}(s)+\beta_{n}(s)+\beta_{n-1}(s)=\frac{z_{n}+\gamma(-1)^{n}}{\beta_{n}(s)}+\delta,
\end{equation}
with
\begin{equation*}
z_{n}=\frac{2n+\alpha}{8Ns},~~\gamma=-\frac{\alpha}{8Ns}~~{\rm and}~~\delta=\frac{s-1}{2s}.
\end{equation*}
\end{pro}
\begin{proof}
Writing
\begin{equation*}
w(x)=|x|^{\alpha}w_{0}(x),
\end{equation*}
where
\begin{equation*}
w_{0}(x):={\rm e}^{-v_{0}(x) }~~~ {\rm with} ~~~v_{0}(x):=Nsx^{4}+N(1-s)x^{2}.
\end{equation*}
With the aid of the three term recurrence relation (\ref{intro2}), repeatedly, we have
\begin{equation*}
\begin{split}
\int_{-\infty}^{\infty}\left[P_{n}(x)P_{n-1}(x)|x|^{\alpha}\right]'w_{0}(x)dx=&4Ns\int_{-\infty}^{\infty}x^{3}P_{n}(x)P_{n-1}(x)w(x)dx\\
&+2N(1-s)\int_{-\infty}^{\infty}xP_{n}(x)P_{n-1}(x)w(x)dx\\
=&4Ns\sqrt{\beta_{n}}\left(\beta_{n+1}+\beta_{n}+\beta_{n-1}\right)+2N(1-s)\sqrt{\beta_{n}},
\end{split}
\end{equation*}
Combining this result and using the identity [\cite{C1}, Eq. (13)],
\begin{equation*}
\int_{-\infty}^{\infty}\left[P_{n}(x)P_{n-1}(x)|x|^{\alpha}\right]'\widetilde{w}(x)dx=\frac{n+\alpha\Delta_{n}}{\sqrt{\beta_{n}}},~~ \Delta_{n}=\frac{1-(-1)^{n}}{2},
\end{equation*}
which holds whenever $\widetilde{w}(x)$ is a symmetric weight on the real line, then we get
\begin{equation}\label{eeee}
\frac{n+\alpha\Delta_{n}}{4Ns}=\beta_{n}\left[\beta_{n+1}+\beta_{n}+\beta_{n-1}+\frac{1-s}{2s}\right].
\end{equation}
\end{proof}
\begin{rem}
Multiplying Eq. (\ref{eeee}) by $s$, we see that $\beta_{n}(0;\alpha,N)=\frac{n+\alpha\Delta_{n}}{2N}$.
\end{rem}
\begin{rem}
The nonlinear discrete equation (\ref{eq3}) can be found from [\cite{C2}, Eq. (23), P. 5], written by Freud, see also [\cite{ANSA}, \S2]. Joshi and Lustri [\cite{JoshiLustri}] studied the large $n$ behavior of ${\rm dP_{I}}$, see also [\cite{PAC3}].
\end{rem}

\begin{pro}\label{pro3.2}
The recurrence coefficients $\beta_{n}(s)$ satisfy the differential-difference equation,
\begin{equation}\label{eq4}
\beta_{n}'(s)=\frac{\beta_{n}(s)}{2s}\left[N(s+1)\left(\beta_{n+1}(s)-\beta_{n-1}(s)\right)-1\right].
\end{equation}
\end{pro}
\begin{proof}
Differentiating $h_{n}(s)$ with respect to $s$,
\begin{equation}\label{pro3.2eq1}
\begin{split}
\frac{dh_{n}(s)}{ds}&=\frac{d}{ds}\int_{-\infty}^{\infty}P_{n}^{2}(x;s)w (x;s)dx\\
&=N\left[\int_{-\infty}^{\infty}\left(x^{2}-x^{4}\right)P_{n}^{2}(x;s)w (x;s)dx\right]\\\
&=Nh_{n}(s)\left[\beta_{n+1}\left(1-\beta_{n+2}-\beta_{n+1}-\beta_{n}\right)+\beta_{n}\left(1-\beta_{n+1}-\beta_{n}-\beta_{n-1}\right)\right]\\
&=Nh_{n}(s)\left[\frac{1+s}{2s}\left(\beta_{n+1}+\beta_{n}\right)-\frac{2n+1}{4Ns}-\frac{\alpha}{4Ns}\right].
\end{split}
\end{equation}
Substituting this into
\begin{equation*}
\begin{split}
\frac{d\beta_{n}(s)}{ds}=\beta_{n}(s)\left[\frac{h'_{n}(s)}{h_{n}(s)}-\frac{h'_{n-1}(s)}{h_{n-1}(s)}\right],
\end{split}
\end{equation*}
the result of (\ref{eq4}) will be obtained.
\end{proof}

Based on the formulas (\ref{eeee}) and (\ref{eq4}), with $n$ and $n-1$, eliminating the term $\beta_{n+1}$ and $\beta_{n-1}$, then we can get a second order ODE for the coefficients $\beta_{n}(s)$.
\begin{lem}
The recurrence coefficients $\beta_{n}(s)$ satisfy the equation
\begin{equation}\label{eq2.5}
\begin{split}
\beta_{n}''(s)&=\frac{\beta_{n}'^{2}(s)}{2\beta_{n}(s)}-\frac{2+s}{s(1+s)}\beta_{n}'(s)+\frac{3N^{2}(1+s)^{2}}{8s^{2}}\beta_{n}^{3}(s)+\frac{N^{2}(1+s)^{2}(1-s)}{4s^{3}}\beta_{n}^{2}(s)\\
&+\bigg[\frac{N^{2}(1-s)^{2}(1+s)^{3}-4s^{2}(3-s)}{32s^{4}(1+s)}
+\frac{N(1+s)^{2}}{16s^{3}}p_{n}\bigg]\beta_{n}(s)-\frac{(1+s)^{2}q_{n}}{128s^{4}\beta_{n}(s)},
\end{split}
\end{equation}
where the parameters $p_{n}$ and $q_{n}$ are given by
\\

$
p_{n}=\left\{\begin{array}{l@{\hspace{1cm}}l}
n+2\alpha &\text{$n$ {\rm even}}\\
n-\alpha&\text{$n$ {\rm odd}}\\
\end{array}\right.$
and ~~$
q_{n}=\left\{\begin{array}{l@{\hspace{1cm}}l}
n^{2}&\text{$n$ {\rm even}}\\
(n+\alpha)^{2}&\text{$n$ {\rm odd}}~.\\
\end{array}\right.
$
\end{lem}
\begin{proof}
From Eq. (\ref{eeee}), we have
\begin{equation}\label{betan3}
\beta_{n-1}=\frac{n+\alpha\Delta_{n}}{4Ns\beta_{n}}-\beta_{n+1}-\beta_{n}-\frac{1-s}{2s},
\end{equation}
and
\begin{equation}\label{betan4}
\beta_{n-2}=\frac{n-1+\alpha\Delta_{n-1}}{4Ns\beta_{n-1}}-\beta_{n}-\beta_{n-1}-\frac{1-s}{2s}.
\end{equation}
From Eq. (\ref{eq4}), we obtain
\begin{equation}\label{betan5}
\beta_{n+1}=\frac{2s}{N(1+s)}\frac{\beta'_{n}}{\beta_{n}}+\frac{1}{N(1+s)}+\beta_{n-1}.
\end{equation}
Substituting Eq. (\ref{betan5}) into Eq. (\ref{betan3}), we get
\begin{equation}\label{betan6}
\beta_{n-1}=-\frac{s}{N(1+s)}\frac{\beta'_{n}}{\beta_{n}}+\frac{n+\alpha\Delta_{n}}{8Ns\beta_{n}}-\frac{\beta_{n}}{2}-\frac{1}{2N(1+s)}-\frac{1-s}{4s}.
\end{equation}
Based on the formulas (\ref{betan5}) and (\ref{betan4}), we find
\begin{equation}\label{betan7}
\begin{split}
\beta_{n}=&\frac{2s}{N(s+1)}\frac{\beta'_{n-1}}{\beta_{n-1}}+\frac{1}{N(1+s)}+\beta_{n-2}\\
=&\frac{2s}{N(1+s)}\frac{\beta'_{n-1}}{\beta_{n-1}}+\frac{1}{N(1+s)}+\frac{n-1+\alpha_{n-1}}{4Ns\beta_{n-1}}-\beta_{n}-\beta_{n-1}-\frac{1-s}{2s}.
\end{split}
\end{equation}
Substituting the expression of $\beta_{n-1}$, given by Eq. (\ref{betan6}), with its derivative into Eq. (\ref{betan7}), one will find the result (\ref{eq2.5}) after some simplification.
\end{proof}

\section{Asymptotics for the recurrence coefficient $\beta_{n}(s)$ as $s\rightarrow0$ }
\begin{thm}\label{thm2}
Assuming $s\rightarrow0$, the recurrence coefficient $\beta_{n}(s;\alpha,N)$ satisfying the second order ODE (\ref{eq2.5}), has the asymptotic expansion
\begin{equation*}
\begin{split}
\beta_{n}=&\frac{n}{2}\bigg\{\frac{1}{N}+\frac{N-3n-2\alpha}{N^{2}}s+\frac{18n^{2}-9nN+N^{2}+22n\alpha+6\left(-N\alpha+\alpha^{2}+1\right)}{N^{3}}s^{2}\\
&+\frac{1}{N^{4}}\Big[N^{3}+nN\left(40n^{2}+\frac{15}{2}n-18N+\frac{85}{2}\right)+30N-135n^{3}-162n-\big(12N^{2}\\
&-30nN+236n^{2}-80n^{2}N+100\big)\alpha
+\big(30N-126n\big)\alpha^{2}-20\alpha^{3}\Big]s^{3}\\
&+\mathcal{O}\big(s^4\big)\bigg\},~~~~n~{\rm even},
\end{split}
\end{equation*}
whilst
\begin{equation*}
\begin{split}
\beta_{n}=&\frac{n}{2}\bigg\{\frac{1}{N}+\frac{N-3n-\alpha}{N^{2}}s+\frac{18n^{2}-9nN+N^{2}+14n\alpha-3N\alpha+2\alpha^{2}+6}{N^{3}}s^{2}\\
&+\frac{1}{N^{4}}\bigg[N^{3}
+nN\left(40n^{2}+\frac{15}{2}n-18N+\frac{85}{2}\right)+30N-135n^{3}-162n\\
&-\big(6N^{2}-40n^{2}N+15nN-\frac{85}{2}N+169n^{2}+62\big)\alpha
+\left(\frac{15}{2}N-40nN-59n\right)\alpha^{2}\\
&-5\alpha^{3}(1+8N)\bigg]s^{3}+\mathcal{O}\big(s^4\big)\bigg\},~~~~n~{\rm odd}.
\end{split}
\end{equation*}
\end{thm}

\begin{proof}
We rewrite the Eq. (\ref{eq2.5}) as follows:
\begin{equation}\label{eq5}
\begin{split}
0&=\frac{1}{s^{4}}\left[\frac{N^{2}(1-s)^{2}(1+s)^{2}\beta_{n}}{32}-\frac{(1+s)^{2}q_{n}}{128\beta_{n}}\right]+\frac{1}{s^{3}}\left[\frac{N^{2}(1+s)^{2}(1-s)}{4}\beta_{n}^{2}
+\frac{N(1+s)^{2}p_{n}}{16}\beta_{n}\right]\\
&+\frac{1}{s^{2}}\left[\frac{3N^{2}(1+s)^{2}}{8}\beta_{n}^{3}-\frac{3-s}{8(1+s)}\beta_{n}\right]-\frac{1}{s}\frac{(2+s)\beta'_{n}}{1+s}+\frac{\beta_{n}'^{2}}{2\beta_{n}}-\beta_{n}''.
\end{split}
\end{equation}
Considering the coefficient of the term $s^{-4}$ then one finds that
\begin{equation*}
\beta_{n}\simeq \frac{\sqrt{q_{n}}}{2N(1-s)}=\frac{\sqrt{q_{n}}}{2N}\left(1+s+s^{2}+s^{3}+\cdots\right).
\end{equation*}
Then we suppose $\beta_{n}$ has the following expansion
\begin{equation*}
\beta_{n}=\frac{\sqrt{q_{n}}}{2N}\left(\sum_{k=0}^{\infty}c_{k}s^{k}\right),~~s\rightarrow0.
\end{equation*}
Substituting this into the formula (\ref{eq2.5}) with directing calculating, we get the asymptotic expansion of $\beta_{n}$, reads
\begin{equation*}
\begin{split}
\beta_{n}(s)=&\frac{\sqrt{q_{n}}}{2N}\bigg[1+\left( \frac{N-p_{n}-2\sqrt{q_{n}}}{N}\right)s+\frac{1}{2N^{2}}\big(2N^{2}-6Np_{n}+3p_{n}^{2}-12N\sqrt{q_{n}}\\
&+16p_{n}\sqrt{q_{n}}+17q_{n}\big)s^{2}+\frac{1}{2N^{3}}\Big(64N+2N^{3}-64p_{n}-12N^{2}p_{n}+15Np_{n}^{2}-5p_{n}^{3} \\
&-128\sqrt{q_{n}}-24N^{2}\sqrt{q_{n}}+80Np_{n}\sqrt{q_{n}}-48p_{n}^{2}\sqrt{q_{n}}+85Nq_{n}-125p_{n}q_{n}\\
&-92q_{n}^{\frac{3}{2}}\Big)s^{3}+\mathcal{O}\left(s^{4}\right)\bigg],~~s\rightarrow0.\\
\end{split}
\end{equation*}
which completes the proof.
\end{proof}

\begin{rem}
For $s=0, ~\alpha=0,~N=1$, which are the Hermite polynomials orthogonal with the weight function ${\rm e}^{-x^{2}}$, one finds $\beta_{n}=\frac{n}{2}$, see [\cite{C2,C1}].
\end{rem}

\begin{figure}[!ht]
\centering
\begin{minipage}[c]{0.5\textwidth}
\centering
\includegraphics[height=4.5cm,width=6.7cm]{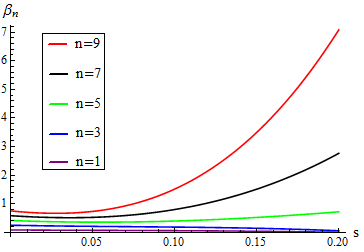}
\end{minipage}%
\begin{minipage}[c]{0.5\textwidth}
\centering
\includegraphics[height=4.5cm,width=6.7cm]{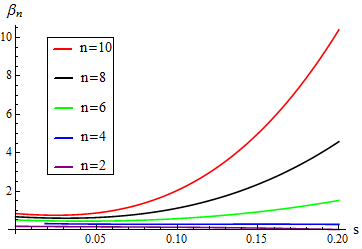}
\end{minipage}
\caption{Plots of the recurrence coefficients $\beta_{n}(s;\alpha,N)$, with $\alpha=12$, $N=6$, $s$ small.}
\end{figure}

\section{Asymptotics for the recurrence coefficient $\beta_{n}(s;\alpha,N)$ as $n\rightarrow\infty$ with $N$ fixed.}
The asymptotic expansion of $\beta_{n}(s;\alpha,N)$, satisfied by Eq. (\ref{eq3}), for the special case when $\alpha=0$ and $s=N=1$ was studied by Lew and Quarles [\cite{Lew}], see also [\cite{Nevai,Noschese}].  To analyze the asymptotic properties of $\beta_{n}(s;\alpha,N)$ as $n\rightarrow\infty$. We first give a brief description of Coulomb fluid, see e.g. [\cite{ChenMckay,ChenLawrence,ChenIsmail2}]. The energy of a system of $n$ logarithmically repelling particles on the line confined by an external potential $v(x)$ reads
\begin{equation*}
E\left(x_{1},x_{2},\ldots,x_{n}\right)=-2\sum_{1\leq j<k\leq n}\log\left|x_{j}-x_{k}\right|+\sum_{j=1}^{n}v\left(x_{j}\right),
\end{equation*}
The collection particles, for large enough $n$, can be approximated as a continuous fluid with a certain density $\sigma(x)$ supported on a single interval $(a,b)\in\mathbb{R}$, see [\cite{Dyson}]. This density, $\sigma(x)$, corresponding to the equilibrium density of the fluid, is obtained by the constrained minimization of the free-energy function, $F[\sigma]$,
\begin{equation*}
F[\sigma]:=\int_{a}^{b}\sigma(x)v(x)dx-\int_{a}^{b}\int_{a}^{b}\sigma(x)\log|x-y|\sigma(y)dxdy.
\end{equation*}
subject to
\begin{equation*}
\int_{a}^{b}\sigma(x)dx=n,~~~\sigma(x)>0.
\end{equation*}

Upon minimization [\cite{Tsuji}], the equilibrium density $\sigma(x)$ is found to satisfy the integral equation,
\begin{equation}\label{eqt7}
L:=v(x)-2\int_{a}^{b}\log|x-y|\sigma(y)dy,~~~~x\in[a,b],
\end{equation}
where $L$ is the Lagrange multiplier. The derivative of this equation over $x$ gives rise to a singular integral equation
\begin{equation*}
v'(x)-2{\rm p. v.}\int_{a}^{b}\frac{\sigma(y)}{x-y}dy=0,~~~x\in[a,b].
\end{equation*}
where ${\rm p. v.}$ denotes the Cauchy principal value. Based on the theory of singular integral equations [\cite{Mikhlin}], we find
\begin{equation}\label{sigmax}
\sigma(x)=\frac{1}{2\pi^{2}}\sqrt{\frac{b-x}{x-a}}{\rm p. v.}\int_{a}^{b}\frac{v'(y)}{y-x}\sqrt{\frac{y-a}{b-y}}dy.
\end{equation}
Thus the normalization, $\int_{a}^{b}\sigma(x)dx=n$, becomes
\begin{equation*}
\frac{1}{2\pi}\int_{a}^{b}\sqrt{\frac{y-a}{b-y}}v'(y)dy=n.
\end{equation*}
with a supplementary condition, see [\cite{ChenIsmail2,ChenLawrence}],
\begin{equation*}
\int_{a}^{b}\frac{v'(x)}{\sqrt{(b-x)(x-a)}}=0.
\end{equation*}
Consequently, the normalization condition becomes
\begin{equation}\label{normalization}
\int_{a}^{b}\frac{xv'(x)}{\sqrt{(b-x)(x-a)}}=2\pi n.
\end{equation}

For the weight at hand, the equilibrium density $\sigma(x)$ supported on $(-b, b)$, and
\begin{equation}\label{vx}
v(x)=-\log w_{\alpha}(x;s)=-\alpha\log|x|+Nsx^{4}+N(1-s)x^{2},
\end{equation}
we find,
\begin{lem}\label{bsquare}
For sufficiently large $n$, with the parameters $N$ and $\alpha$ finite, we have
\begin{equation}\label{b}
b^{2}\simeq\frac{2\sqrt{3}}{3}N^{-\frac{1}{2}}s^{-\frac{1}{2}}n^{\frac{1}{2}}.
\end{equation}
\end{lem}

\begin{proof}
It follows from Eq. (\ref{vx}) that
\begin{equation*}
v'(x)=-{\rm p. v.}\frac{\alpha}{x}+4Nsx^{3}+2N(1-s)x.
\end{equation*}
Putting $a=-b$, we find from Eq. (\ref{normalization}) that
\begin{equation}\label{bb}
\begin{split}
2\pi n=&\lim_{\varepsilon\rightarrow0^{+}}\left(\int_{-b}^{-\varepsilon}+\int_{\varepsilon}^{b}\right)\frac{-\alpha+4Nsx^{4}+2N(1-s)x^{2}}{\sqrt{b^{2}-x^{2}}}dx\\
=&\lim_{\varepsilon\rightarrow0^{+}}\int_{\varepsilon^{2}}^{b^{2}}\frac{-\alpha+4Nsy^{2}+2N(1-s)y}{\sqrt{\left(b^{2}-y\right)y}}dy\\
=&-\alpha\pi+\frac{3\pi}{2}Nsb^{4}+\pi N(1-s) b^{2}.
\end{split}
\end{equation}
Since $b^{2}>0$, we have
\begin{equation}\label{b2}
\begin{split}
b^{2}=&\frac{s-1+\sqrt{(1-s)^{2}+6s\left(\frac{2n+\alpha}{N}\right)}}{3s}\\
\simeq&\frac{2\sqrt{3}}{3}N^{-\frac{1}{2}}s^{-\frac{1}{2}}n^{\frac{1}{2}},~~n\rightarrow\infty.
\end{split}
\end{equation}
\end{proof}

\begin{lem}\label{upperbound}
A upper bound for the recurrence coefficient $\beta_{n}(s;\alpha,N)$ with respect to the weight $w_{\alpha}(x;s,N)$ is given by
\begin{equation}\label{b}
\beta_{n}<\frac{n^{\frac{1}{2}}}{2\sqrt{Ns}}+\frac{s-1}{4s}+\frac{\alpha\Delta_{n}+\frac{N(1-s)^{2}}{4s}}{4\sqrt{Ns}}n^{-\frac{1}{2}}
+\sum_{j=2}^{\infty}\frac{b_{j}}{\sqrt{Ns}}\frac{\left[\alpha\Delta_{n}+\frac{N(1-s)^{2}}{4s}\right]^{j}}{n^{j-\frac{1}{2}}},
\end{equation}
where
\begin{equation*}
b_{j}=\frac{(-1)^{j-1}(2j-3)!}{2^{2j-1}j!(j-2)!}.
\end{equation*}
\end{lem}
\begin{proof}
From Eq. (\ref{eq3}), we obtain
\begin{equation*}
\beta_{n}^{2}+\frac{1-s}{2s}\beta_{n}+\beta_{n}\left(\beta_{n+1}+\beta_{n-1}\right)=\frac{n+\alpha\Delta_{n}}{4Ns},
\end{equation*}
it follows that
\begin{equation*}
\beta_{n}^{2}+\frac{1-s}{2s}\beta_{n}<\frac{n+\alpha\Delta_{n}}{4Ns},
\end{equation*}
due to $\beta_{n}=\frac{h_{n}}{h_{n-1}}>0$, where
\begin{equation*}
h_{n}(s;\alpha,N)=\int_{-\infty}^{\infty}P_{n}^2(x)w_{\alpha}(x;s,N)dx.
\end{equation*}
So
\begin{equation*}
\begin{split}
0<\beta_{n}&<-\frac{1-s}{4s}+\sqrt{\left(\frac{1-s}{4s}\right)^{2}+\frac{n+\alpha\Delta_{n}}{4Ns}}\\
&=\frac{s-1}{4s}
+\frac{n^{\frac{1}{2}}}{2\sqrt{Ns}}\left[1+\frac{\alpha\Delta_{n}+\frac{N(1-s)^{2}}{4s}}{n}\right]^{\frac{1}{2}}\\
&=\frac{n^{\frac{1}{2}}}{2\sqrt{Ns}}+\frac{s-1}{4s}+\frac{\alpha\Delta_{n}+\frac{N(1-s)^{2}}{4s}}{4\sqrt{Ns}}n^{-\frac{1}{2}}
+\sum_{j=2}^{\infty}\frac{b_{j}}{\sqrt{Ns}}\frac{\left[\alpha\Delta_{n}+\frac{N(1-s)^{2}}{4s}\right]^{j}}{n^{j-\frac{1}{2}}},
\end{split}
\end{equation*}
where
\begin{equation*}
b_{j}=\frac{(-1)^{j-1}(2j-3)!!}{2^{j+1}j!}=\frac{(-1)^{j-1}(2j-3)!}{2^{2j-1}j!(j-2)!},
\end{equation*}
which completes the proof.
\end{proof}

In the following, we provide the asymptotic expansion of $\beta_{n}(s;\alpha,N)$ in Eq. (\ref{intro2}) as $n\rightarrow\infty$ with $N$ fixed, for $s\in[0,1],\alpha>-1$.

\begin{thm}\label{thm1}
Let $n\rightarrow\infty$ with $N$ fixed, the recurrence coefficient $\beta_{n}(s;\alpha,N)$ associated with monic deformed Freud polynomials satisfying the nonlinear discrete equation (\ref{eq3}), i.e.
\begin{equation}\label{eq3.1}
\beta_{n}\left(\beta_{n+1}+\beta_{n}+\beta_{n-1}+\frac{1-s}{2s}\right)=\frac{n+\alpha\Delta_{n}}{4Ns},~~~\Delta_{n}=\frac{1-(-1)^{n}}{2},
\end{equation}
has the asymptotic expansion
\begin{equation}\label{eq3.2}
\begin{split}
\beta_{n}=&\frac{1}{2\sqrt{3s}}\sqrt{\frac{n}{N}}-\frac{1-s}{12s}+\frac{(1-s)^{2}}{48\sqrt{3}s^{\frac{3}{2}}}\sqrt{\frac{N}{n}}+\left[\frac{1}{48\sqrt{3Ns}}
-\frac{N^{\frac{3}{2}}(1-s)^{4}}{2304\sqrt{3}s^{\frac{5}{2}}}\right]n^{-\frac{3}{2}}\\
&-\frac{1-s}{288s}n^{-2}+\frac{N^{\frac{5}{2}}(1-s)^{6}-144\sqrt{N}s^{2}(1-s)^{2}}{55296\sqrt{3}s^{3}}n^{-\frac{5}{2}}+\frac{N(1-s)^{3}}{1728s^{2}}n^{-3}\\
&+\mathcal{O}\left(n^{-\frac{7}{2}}\right),~~~\text{$n$ {\rm even}},
\end{split}
\end{equation}
whilst
\begin{equation}\label{eq3.3}
\begin{split}
\beta_{n}=&\frac{1}{2\sqrt{3s}}\sqrt{\frac{n}{N}}-\frac{1-s}{12s}+\left[\frac{\sqrt{N}(1-s)^{2}}{48\sqrt{3}s^{\frac{3}{2}}}+\frac{\alpha}{4\sqrt{3Ns}}\right]\frac{1}{\sqrt{n}}
+\bigg\{-\frac{\left[N(1-s)^{2}+12s\alpha\right]^{2}}{2304\sqrt{3N}s^{\frac{5}{2}}}\\
&+\frac{1}{48\sqrt{3Ns}}\bigg\}n^{-\frac{3}{2}}-\frac{1-s}{288s}n^{-2}+\Psi(s)n^{-\frac{5}{2}}+\frac{N(1-s)^{3}+12s(1-s)\alpha}{1728s^{2}}n^{-3}\\
&+\mathcal{O}\left(n^{-\frac{7}{2}}\right),~~~\text{$n$ {\rm odd}},
\end{split}
\end{equation}
with
\begin{equation*}
\Psi(s)=\frac{N^{\frac{5}{2}}(1-s)^{6}+36N^{\frac{3}{2}}(1-s)^{4}s\alpha+144\sqrt{N}(1-s)^{2}s^{2}\left(3\alpha^{2}-1\right)+1728s^{3}\alpha\left(\alpha^{2}-1\right)}
{55296\sqrt{3}s^{\frac{7}{2}}}.
\end{equation*}
\end{thm}
\begin{proof}
In fact, Lemma \ref{bsquare} gives the first term in the asymptotic expansion of $\beta_{n}(s;\alpha,N)$, since
\begin{equation*}
\beta_{n}\simeq\left(\frac{b-a}{4}\right)^{2}\simeq\frac{1}{2\sqrt{3Ns}}n^{\frac{1}{2}},~a=-b.
\end{equation*}
Hence it seems reasonable to assume $\beta_{n}(s;\alpha,N)$ has the expansion of the form
\begin{equation}\label{eq3.8}
\beta_{n}=\frac{n^{\frac{1}{2}}}{2\sqrt{3Ns}}\left(1+\sum_{k=1}^{\infty}d_{k}n^{-\frac{k}{2}}\right),~~n\rightarrow\infty.
\end{equation}
Replacing $n$ by $n\pm1$ in Eq. (\ref{eq3.8}), we obtain
\begin{equation}\label{eq3.9}
\begin{split}
\beta_{n\pm1}=&\frac{(n\pm1)^{\frac{1}{2}}}{2\sqrt{3Ns}}\left[1+\sum_{k=1}^{\infty}d_{k}(n\pm1)^{-\frac{k}{2}}\right]\\
=&\frac{n^{\frac{1}{2}}\left(1\pm \frac{1}{n}\right)^{\frac{1}{2}}}{2\sqrt{3Ns}}\left[1+\sum_{k=1}^{\infty}d_{k}n^{-\frac{k}{2}}\left(1\pm\frac{1}{n}\right)^{-\frac{k}{2}}\right]\\
=&\frac{n^{\frac{1}{2}}}{2\sqrt{3Ns}}\bigg[1+\frac{d_{1}}{n^{\frac{1}{2}}}+\frac{2d_{2}\pm1}{2n}+\frac{d_{3}}{n^{\frac{3}{2}}}+\frac{8d_{4}\mp4d_{2}-1}{8n^{2}}
+\frac{d_{5}\mp d_{3}}{n^{\frac{5}{2}}}\\
&+\frac{16d_{6}\mp24d_{4}+6d_{2}\pm1}{16n^{3}}+\frac{d_{7}\mp2d_{5}+d_{3}}{n^{\frac{7}{2}}}+\mathcal{O}\left(n^{-4}\right)\bigg],
\end{split}
\end{equation}
by doing an asymptotic expansion in the above equation. Substituting Eqs. (\ref{eq3.8}) and (\ref{eq3.9}) into Eq. (\ref{eq3.1}), followed by comparing the corresponding coefficients on both sides, we find
\begin{equation}\label{eq3.999}
\begin{split}
&d_{1}=-\frac{\sqrt{N}(1-s)}{2\sqrt{3s}},~~d_{2}=\frac{N(1-s)^{2}}{24s}+\frac{\alpha\Delta_{n}}{2},~~d_{3}=0,\\
&d_{4}=\frac{1}{24}-\frac{\left[N(1-s)^{2}+12s\alpha\Delta_{n}\right]^{2}}{1152s^{2}},~~d_{5}=-\frac{\sqrt{N}(1-s)}{48\sqrt{3s}},\\
&d_{6}=\frac{\left[N(1-s)^{2}+12s\alpha\Delta_{n}\right]\left[N(1-s)^{2}-12s+12s\alpha\Delta_{n}\right]\left[N(1-s)^{2}+12s+12s\alpha\Delta_{n}\right]}{27648s^{3}},\\
&d_{7}=\frac{\sqrt{N}(1-s)\left[N(1-s)^{2}+12s\alpha\Delta_{n}\right]}{288\sqrt{3}s^{\frac{3}{2}}}.
\end{split}
\end{equation}
Hence, we obtain the asymptotic expansions (\ref{eq3.2}) and (\ref{eq3.3}), followed by some computation.
\end{proof}

\begin{cor}
Assume that $\beta_{n}(s;\alpha,N)$ satisfying Eq. (\ref{eq3.1}). Then for $s\in[0,1], ~\alpha>-1,~N>0$:

\noindent{\rm(\textrm{i})} the sequence $\left\{\frac{\beta_{n}(s;\alpha,N)}{\sqrt{n}}\right\}_{n=1}^{\infty}$ is bounded;

\noindent{\rm(\textrm{ii})} $\lim_{n\rightarrow\infty}\frac{\beta_{n}(s;\alpha,N)}{\sqrt{n}}=\frac{1}{2\sqrt{3Ns}}$.
\end{cor}

\begin{rem}
Putting $s=1,~\alpha=0,~N=1$, the classical result, obtained by Lew and Quarles [\cite{Lew}] for the Freud weight ${\rm e}^{-x^{4}}$, is recovered, i.e.
\begin{equation*}
\lim_{n\rightarrow\infty}\frac{\beta_{n}(1;0,1)}{\sqrt{n}}=\frac{1}{2\sqrt{3}}.
\end{equation*}
\end{rem}

\begin{cor}\label{thm3}
For $s\in[0,1], ~\alpha>-1,~N>0$, the recurrence coefficients $\beta_{n}(s;\alpha,N)$ in Eq. (\ref{eq3.1}) satisfy
\begin{equation}\label{eq4.9}
\begin{split}
\frac{\beta_{n+1}(s;\alpha,N)}{\beta_{n}(s;\alpha,N)}&=1+\mathcal{O}(n^{-1}),~~~~\text{$n\rightarrow\infty$},\\
\frac{\beta_{n}(s;\alpha,N)}{a_{\mu}^{2}}&=\frac{1}{4}+\mathcal{O}(n^{-1}),~~~~\text{$n\rightarrow\infty$}.
\end{split}
\end{equation}
where $a_{\mu}$ is sometimes called the Mhaskar-Rakhmanov-Saff number [\cite{C7,C6}], the unique positive root of the equation
\begin{equation*}
\mu=\frac{2}{\pi}\int_{0}^{1}\frac{a_{\mu} y Q'(a_{\mu}y)}{\sqrt{1-y^{2}}}dy,
\end{equation*}
with $Q(x)=Nsx^{4}+N(1-s)x^{2}$.
\end{cor}
\begin{proof}
See [\cite{C5}, Thm. 2.1].
\end{proof}

\section{Asymptotics for the recurrence coefficient $\beta_{n}(s;\alpha,N)$ as $n\rightarrow\infty, N\rightarrow\infty$ with $\frac{n}{N}$ fixed.}
In Section 2 we have derived that the recurrence coefficient $\beta_{n}(s;\alpha,N)$ satisfies the second order ODE (\ref{eq2.5}). Let $0<r_{0}\leq r:=\frac{n}{N}\leq 1$, i.e. the quantity $\frac{n}{N}$ is bounded away from $0$ and close to $1$. Substituting $n=Nr$ into Eq. (\ref{eq2.5}) gives,
\begin{equation}\label{eq2.55}
\begin{split}
&\bigg[\frac{3(1+s)^{2}}{8s^{2}}\beta_{n}^{3}(s)+\frac{(1+s)^{2}(1-s)}{4s^{3}}\beta_{n}^{2}(s)+\frac{(1+s)^{2}(1-2s+2rs+s^{2})}{32s^{4}}\beta_{n}(s)\\
&-\frac{(1+s)^{2}r^{2}}{128s^{4}\beta_{n}(s)}\bigg]N^{2}+\left[\frac{\alpha(1+s)^{2}}{8s^{3}}\beta_{n}(s)\right]N-\beta_{n}''(s)+\frac{\beta_{n}'^{2}(s)}{2\beta_{n}(s)}
-\frac{2+s}{s(1+s)}\beta_{n}'(s)\\
&-\frac{3-s}{8s^{2}(1+s)}\beta_{n}(s)=0,~~~\text{$n$ is even},
\end{split}
\end{equation}
and
\begin{equation}\label{eq2.555}
\begin{split}
&\bigg[\frac{3(1+s)^{2}}{8s^{2}}\beta_{n}^{3}(s)+\frac{(1+s)^{2}(1-s)}{4s^{3}}\beta_{n}^{2}(s)+\frac{(1+s)^{2}(1-2s+2rs+s^{2})}{32s^{4}}\beta_{n}(s)\\
&-\frac{(1+s)^{2}r^{2}}{128s^{4}\beta_{n}(s)}\bigg]N^{2}-\left[\frac{\alpha(1+s)^{2}}{16s^{3}}\beta_{n}(s)+\frac{r\alpha(1+s)^{2}}{64s^{4}\beta_{n}(s)}\right]N-\beta_{n}''(s)
+\frac{\beta_{n}'^{2}(s)}{2\beta_{n}(s)}\\
&-\frac{2+s}{s(1+s)}\beta_{n}'(s)
-\frac{3-s}{8s^{2}(1+s)}\beta_{n}(s)-\frac{\alpha^{2}(1+s)^{2}}{128s^{4}\beta_{n}(s)}=0,~~~\text{$n$ is odd},
\end{split}
\end{equation}
which imply that
\begin{equation}\label{eq2.5555}
\begin{split}
&\frac{3(1+s)^{2}}{8s^{2}}\beta_{n}^{3}(s)+\frac{(1+s)^{2}(1-s)}{4s^{3}}\beta_{n}^{2}(s)+\frac{(1+s)^{2}(1-2s+2rs+s^{2})}{32s^{4}}\beta_{n}(s)\\
&-\frac{(1+s)^{2}r^{2}}{128s^{4}\beta_{n}(s)}=0, ~~~\text{$N\rightarrow\infty$}.
\end{split}
\end{equation}
Solving this equation, we get the solution
\begin{equation}\label{eq2.55555}
\begin{split}
\beta_{n}(s)=\frac{s-1+\sqrt{1-2s+12rs+s^{2}}}{12s},
\end{split}
\end{equation}
which is consistent with the result (\ref{b2}) since
\begin{equation}\label{eq2.556}
\begin{split}
\beta_{n}(s)\simeq\frac{b^{2}}{4},~N\rightarrow\infty.
\end{split}
\end{equation}
Therefore, we assume that $\beta_{n}(s)$ has the asymptotic expansion
\begin{equation}\label{eq2.557}
\begin{split}
\beta_{n}(s)= a_{0}(s)+\sum_{k=1}^{\infty}\frac{a_{k}(s)}{N^{k}},~N\rightarrow\infty.
\end{split}
\end{equation}
where
\begin{equation*}
\begin{split}
a_{0}(s):=\frac{s-1+\sqrt{1-2s+12rs+s^{2}}}{12s}.
\end{split}
\end{equation*}
Let $g(s):=\sqrt{1-2s+12rs+s^{2}}$ and $f(s):=1-2s-4rs+s^{2}$. Then substituting the asymptotic expansion Eq. (\ref{eq2.557}) into Eqs. (\ref{eq2.55}) and (\ref{eq2.555}), respectively, we obtain
\begin{equation*}
a_{1}(s)=
\begin{cases}
\frac{\left[s-1+g(s)\right]\alpha}{4(s-1)g(s)-2g^{2}(s)}, ~~~~~~~~n~{\rm  even}, \\ \\
\frac{\alpha(1-s)}{2f(s)}-\frac{4\alpha s r}{f(s)g(s)},\ \
~~~~ ~n~{\rm  odd}.
\end{cases}
\end{equation*}
and, for $n$ is even
\begin{equation*}
\begin{split}
a_{2}(s)=\frac{s(s-1)}{2g^{4}(s)}+\frac{s\left(4-3\alpha^{2}\right)}{8g^{3}(s)}+\frac{3\alpha^{2}s(s-1)}{2f^{2}(s)}+\frac{15\alpha^{2}s+3\alpha^{2}s^{2}(12r+5s-10)}{8f(s)g(s)},
\end{split}
\end{equation*}
else, if $n$ is odd,
\begin{equation*}
\begin{split}
a_{2}(s)=\frac{s(s-1)}{2g^{4}(s)}+\frac{s(4-3\alpha^{2})}{8g^{3}(s)}+\frac{\alpha^{2}s(1-s)}{2f^{2}(s)}-\frac{3\alpha^{2}s}{8f(s)g(s)}.
\end{split}
\end{equation*}
\begin{thm}\label{nN}
As $n, N\rightarrow\infty$, with the quantity $\frac{n}{N}$ is bounded away from $0$ and close to $1$, the recurrence coefficients have the asymptotic expansion
\begin{equation*}
\begin{split}
\beta_{n}(s)=&\frac{s-1+g(s)}{12s}+\frac{\left[s-1+g(s)\right]\alpha}{4(s-1)g(s)-2g^{2}(s)}\frac{1}{N}+\bigg[\frac{s(s-1)}{2g^{4}(s)}+\frac{s\left(4-3\alpha^{2}\right)}{8g^{3}(s)}
\\
&+\frac{3\alpha^{2}s(s-1)}{2f^{2}(s)}+\frac{15\alpha^{2}s+3\alpha^{2}s^{2}(12r+5s-10)}{8f(s)g(s)}\bigg]\frac{1}{N^{2}}+\mathcal{O}\left(N^{-3}\right),~\text{$n$ {\rm even}},
\end{split}
\end{equation*}
whilst
\begin{equation*}
\begin{split}
\beta_{n}(s)=&\frac{s-1+g(s)}{12s}+\left[\frac{\alpha(1-s)}{2f(s)}-\frac{4\alpha s r}{f(s)g(s)}\right]\frac{1}{N}+\bigg[\frac{s(s-1)}{2g^{4}(s)}+\frac{s(4-3\alpha^{2})}{8g^{3}(s)}\\
&+\frac{\alpha^{2}s(1-s)}{2f^{2}(s)}-\frac{3\alpha^{2}s}{8f(s)g(s)}\bigg]\frac{1}{N^{2}}+\mathcal{O}\left(N^{-3}\right),~~~~~~~~~~~~~~~~~~~~~~~~~~\text{$n$ {\rm odd}},
\end{split}
\end{equation*}
with
\begin{equation*}
\begin{split}
g(s)=\sqrt{1-2s+12rs+s^{2}}~~~{\rm and}~~~f(s)=1-2s-4rs+s^{2}.
\end{split}
\end{equation*}
\end{thm}

\begin{rem}
Higher order corrections can be calculated systematically in Mathematica, but will be quite cumbersome. Hence we only give the expansion terms up to order $N^{-2}$.
\end{rem}

\section{The deformed Freud polynomials}
 Based on an extension of the ladder operators technique developed in [\cite{ladder operator}], Chen and Feigin [\cite{ChenFeigin}] studied the weight $\widetilde{w}(x)|x-t|^{K}, x,t,K\in\mathbb{R}$, for any smooth reference weight $\widetilde{w}(x)$. They showed that when $\widetilde{w}(x)$ is the Gaussian (Hermite) weight (${\rm e}^{-x^{2}}, x\in\mathbb{R}$), the recurrence coefficients satisfy a particular two-parameter Painlev\'{e} IV equation. Filipuk \textit{et al} [\cite{Filipuk}] found that the recurrence coefficients for the Freud weight $|x|^{2\alpha+1}{\rm e}^{-x^{4}+tx^{2}}, x,t\in\mathbb{R}, \alpha>-1$ are related to the solutions of the Painlev\'{e} IV and the first discrete Painlev\'{e} equation. In [\cite{PAC2,PAC3}], Clarkson \textit{et al} gave a systematic study on Freud weight $|x|^{2\alpha+1}{\rm e}^{-x^{4}+tx^{2}}$, and some generalized work for [\cite{ChenFeigin}]. In the following theorem, we give the differential-difference equation (i.e., lowing operator) satisfied by the deformed Freud orthogonal polynomials associated with the weight $w_{\alpha}(x;s,N)$ given in Eq. (\ref{weight}).
\begin{lem}\label{lem5.1}
The monic orthogonal polynomials $P_{n}(z;s,\alpha,N)$ with respect to the deformed Freud weight $w_{\alpha}(x;s,N)$ on $\mathbb{R}$ satisfy the differential-differentce recurrence relation
\begin{equation}\label{differential-difference}
P'_{n}(z)=\beta_{n}(s)A_{n}(z)P_{n-1}(z)-B_{n}(z)P_{n}(z),
\end{equation}
where
\begin{equation*}
\begin{split}
&A_{n}(z):=\frac{1}{h_{n}}\int_{-\infty}^{\infty}\frac{v_{0}'(z)-v_{0}'(y)}{z-y}P^{2}_{n}(y)w(y)dy,\\
&B_{n}(z):=\frac{1}{h_{n-1}}\int_{-\infty}^{\infty}\frac{v_{0}'(z)-v_{0}'(y)}{z-y}P_{n}(y)P_{n-1}(y)w(y)dy+\frac{\alpha\left[1-(-1)^{n}\right]}{2z},
\end{split}
\end{equation*}
with
\begin{equation}\label{vy0}
v_{0}(x)=Nsx^{4}+N(1-s)x^{2},~x\in\mathbb{R},~s\in[0,1],~N>0.
\end{equation}
\end{lem}
\begin{proof}
Since the derivative of $P_{n}(z)$ is a polynomial of degree $n-1$ in $z$, hence $P'_{n}(z)$ can be given by
\begin{equation}\label{1412-1}
P'_{n}(z)=\sum_{j=0}^{n-1}c_{n,j}P_{j}(z).
\end{equation}
Using the orthogonality relations, and integrating by parts, we have
\begin{equation}\label{1412-2}
c_{n,j}=\frac{1}{h_{j}}\int_{-\infty}^{\infty}P'_{n}(y)P_{j}(y)w(y)dy=\frac{1}{h_{j}}\int_{-\infty}^{\infty}P_{n}(y)P_{j}(y)\left[v_{0}'(y)-\frac{\alpha}{y}\right]w(y)dy,
\end{equation}
where
\begin{equation}\label{vy}
v_{0}(y):=Nsy^{4}+N(1-s)y^{2}.
\end{equation}
Substituting Eq. (\ref{1412-2}) into Eq. (\ref{1412-1}) and summating over $j$ applying the Christoffel-Darboux formula [\cite{Ismail}, Theorem 2.2.2],
\begin{equation*}
\sum_{j=0}^{n-1}\frac{P_{j}(z)P_{j}(y)}{h_{j}}=\frac{P_{n}(z)P_{n-1}(y)-P_{n}(y)P_{n-1}(z)}{(z-y)h_{n-1}},
\end{equation*}
we have
{\small\begin{equation}\label{zuichang}
\begin{split}
P'_{n}(z)=&\int_{-\infty}^{\infty}P_{n}(y)\sum_{j=0}^{n-1}\frac{P_{j}(z)P_{j}(y)}{h_{j}}\left[v_{0}'(y)-\frac{\alpha}{y}\right]w(y)dy\\
=&\int_{-\infty}^{\infty}P_{n}(y)\sum_{j=0}^{n-1}\frac{P_{j}(z)P_{j}(y)}{h_{j}}\left[v_{0}'(y)-v_{0}'(z)\right]w(y)dy\\
&-\alpha\int_{-\infty}^{\infty}\frac{P_{n}(y)}{y}\sum_{j=0}^{n-1}\frac{P_{j}(z)P_{j}(y)}{h_{j}}w(y)dy\\
=&-\frac{1}{h_{n-1}}\int_{-\infty}^{\infty}\left[P_{n}(z)P_{n}(y)P_{n-1}(y)-P^{2}_{n}(y)P_{n-1}(z)\right]\frac{v_{0}'(z)-v_{0}'(y)}{z-y}w(y)dy\\
&-\frac{\alpha}{z}\int_{-\infty}^{\infty}\frac{[(z-y)+y]P_{n}(y)}{y}\left[\sum_{j=0}^{n-1}\frac{P_{j}(z)P_{j}(y)}{h_{j}}\right]w(y)dy\\
=&\frac{P_{n-1}(z)}{h_{n-1}}\int_{-\infty}^{\infty}P^{2}_{n}(y)\frac{v_{0}'(z)-v_{0}'(y)}{z-y}w(y)dy-\frac{\alpha}{z}\frac{ P_{n}(z)}{h_{n-1}}\int_{-\infty}^{\infty}\frac{P_{n}(y)P_{n-1}(y)}{y}w(y)dy\\
&-\frac{P_{n}(z)}{h_{n-1}}\int_{-\infty}^{\infty}P_{n}(y)P_{n-1}(y)\frac{v_{0}'(z)-v'_{0}(y)}{z-y}w(y)dy.
\end{split}
\end{equation}}
Furthermore, by an inductive argument based on the three term recurrence relation Eq. (\ref{intro2}) with the initial conditions
\begin{equation}\label{initialcondition}
P_{0}(z):=1, ~~~{\rm and}~~~\beta_{0}P_{-1}(z)=0,
\end{equation}
we obtain
\begin{equation}\label{bn}
\begin{split}
\frac{1}{h_{n-1}}\int_{-\infty}^{\infty}\frac{P_{n}(y)P_{n-1}(y)}{y}w(y)dy&=\frac{1}{h_{n-1}}\int_{-\infty}^{\infty}\frac{\left[yP_{n-1}(y)-\beta_{n-1}P_{n-2}(y)\right]P_{n-1}(y)}{y}w(y)dy\\
&=\begin{cases}
1, ~~~~~~~~~~~~~~~~~~~~~~~~~~~~~~~~~~~~~~~~~~~~n=1, \\
1-\frac{1}{h_{n-2}}\int_{-\infty}^{\infty}\frac{P_{n-1}(y)P_{n-2}(y)}{y}w(y)dy,~~n\geq2,\ \
\end{cases}\\
&=\begin{cases}
0, ~~~~~~~~~~~~~~~~~~~~~~~~~~~~~~~~~~~~~~~~~~~~n=2, \\
\frac{1}{h_{n-3}}\int_{-\infty}^{\infty}\frac{P_{n-2}(y)P_{n-3}(y)}{y}w(y)dy,~~~~~~~~n\geq3,\ \
\end{cases}\\
&\ldots\\
&=\begin{cases}
0, ~~~~~~~~~~~~~~~~~~~~~~~~~~~~~~~~~~~~~~~~~~~~~n~{\rm even}, \\
1,~~~~~~~~~~~~~~~~~~~~~~~~~~~~~~~~~~~~~~~~~~~~~n~{\rm odd}.\ \
\end{cases}\\
\end{split}
\end{equation}
Therefore, (bearing in mind $\beta_{n}=h_{n}/h_{n-1}$) the result is derived by Eqs. (\ref{zuichang}) and (\ref{bn}) immediately.
\end{proof}

\begin{lem}\label{lem5.2}
$A_{n}(z)$ and $B_{n}(z)$ defined by Lemma \ref{lem5.1} satisfy the relation:
\begin{equation}\label{AnBn}
A_{n}(z)=\frac{v_{0}'(z)}{z}+\frac{B_{n}(z)+B_{n+1}(z)}{z}-\frac{\alpha}{z^{2}}.
\end{equation}
\end{lem}
\begin{proof}
Be the definition of $A_{n}(z)$, we rewrite it as
{\small\begin{equation*}
\begin{split}
A_{n}(z)=&\frac{1}{zh_{n}}\left\{\int_{-\infty}^{\infty}\frac{v_{0}'(z)-v_{0}'(y)}{z-y}yP_{n}^{2}(y)w(y)dy+\int_{-\infty}^{\infty}\left[v_{0}'(z)-v_{0}'(y)\right]P_{n}^{2}(y)w(y)dy\right\}\\
=&\frac{1}{zh_{n}}\bigg\{\int_{-\infty}^{\infty}\frac{v_{0}'(z)-v_{0}'(y)}{z-y}\left[P_{n+1}(y)+\beta_{n}P_{n-1}(y)\right]P_{n}(y)w(y)dy+v_{0}'(z)h_{n}\\
&-\int_{-\infty}^{\infty}P_{n}^{2}(y)\left[\frac{\alpha}{y}w(y)-w'(y)\right]dy\bigg\}\\
=&\frac{v_{0}'(z)}{z}+\frac{1}{z}\left\{B_{n+1}(z)-\frac{\alpha}{2z}\left[1-(-1)^{n+1}\right]+B_{n}(z)-\frac{\alpha}{2z}\left[1-(-1)^{n}\right]\right\}\\
=&\frac{v_{0}'(z)}{z}+\frac{B_{n}(z)+B_{n+1}(z)}{z}-\frac{\alpha}{z^{2}},
\end{split}
\end{equation*}}
which completes the proof.
\end{proof}

\begin{thm}\label{thm5.1}
The monic orthogonal polynomials $P_{n}(z;s,\alpha,N)$ with respect to deformed Freud weight (\ref{weight}) satisfy the linear second order ODE:
\begin{equation}\label{newnew1}
P''_{n}(z)+R_{n}(z)P'_{n}(z)+Q_{n}(z)P_{n}(z)=0,
\end{equation}
where{\small
\begin{equation}\label{newnew2}
R_{n}(z)=-4Nsz^{3}-2N(1-s)z+\frac{\alpha}{z}-\frac{2z}{z^{2}+\frac{1-s}{2s}+\beta_{n}+\beta_{n+1}},
\end{equation}
\begin{equation}\label{newnew3}
\begin{split}
Q_{n}(z)=&4Ns\beta_{n}\left[1+\alpha(-1)^{n}\right]+16N^{2}s^{2}\beta_{n}\left[\frac{1-s}{2s}+\beta_{n}+\beta_{n-1}\right]\left[\frac{1-s}{2s}+\beta_{n}+\beta_{n+1}\right]\\
&-\alpha\left[1-(-1)^{n}\right]\left[N(1-s)+\frac{1}{2z^{2}}\right]-\frac{8Ns\beta_{n}z^{2}+\alpha\left[1-(-1)^{n}\right]}{z^{2}+\frac{1-s}{2s}+\beta_{n}+\beta_{n+1}}+4nNsz^{2}.\\
\end{split}
\end{equation}}
\end{thm}
\begin{proof}
From the differential-difference equation (\ref{differential-difference}), we have the following type of rasing operator,
\begin{equation*}
\begin{split}
P_{n-1}'(z)=&-B_{n-1}(z)P_{n-1}(z)+A_{n-1}(z)\beta_{n-1}(s)P_{n-2}(z)\\
=&-B_{n-1}(z)P_{n-1}(z)+zA_{n-1}(z)P_{n-1}(z)-A_{n-1}(z)P_{n}(z)\\
\end{split}
\end{equation*}
Using the relation between $A_{n}(z)$ and $B_{n}(z)$, see Lemma \ref{lem5.2}, we get
{\small \begin{equation}\label{useforsimplify}
\begin{split}
P_{n-1}'(z)=&-B_{n-1}(z)P_{n-1}(z)+z\left[\frac{v_{0}'(z)}{z}+\frac{B_{n-1}(z)+B_{n}(z)}{z}-\frac{\alpha}{z^{2}}\right]P_{n-1}(z)-A_{n-1}(z)P_{n}(z)\\
=&\left[v_{0}'(z)+B_{n}(z)-\frac{\alpha}{z}\right]P_{n-1}(z)-A_{n-1}(z)P_{n}(z).\\
\end{split}
\end{equation}}

Differentiating both sides of Eq. (\ref{differential-difference}) with respect to $z$ we obtain
\begin{equation*}
\begin{split}
P_{n}''(z)=&-B_{n}'(z)P_{n}(z)-B_{n}(z)P_{n}'(z)+\beta_{n}(s)A_{n}'(z)P_{n-1}(z)+\beta_{n}(s)A_{n}(z)P_{n-1}'(z).\\
\end{split}
\end{equation*}
Substituting Eq. (\ref{useforsimplify}) into above, we obtain
\begin{equation}\label{pn''}
\begin{split}
P_{n}''(z)=&-B_{n}'(z)P_{n}(z)-B_{n}(z)P_{n}'(z)-\beta_{n}(s)A_{n}(z)A_{n-1}(z)P_{n}(z)\\
&+\left\{\beta_{n}(s)A_{n}'(z)+\beta_{n}(s)A_{n}(z)\left[v_{0}'(z)+B_{n}(z)-\frac{\alpha}{z}\right]\right\}P_{n-1}(z).\\
\end{split}
\end{equation}
Eliminating $P_{n-1}(z)$ from Eqs. (\ref{differential-difference}) and (\ref{pn''}), we find
{\small\begin{equation}\label{pn'''}
\begin{split}
P_{n}''(z)+R_{n}(z)P_{n}'(z)+Q_{n}(z)P_{n}(z)=0,
\end{split}
\end{equation}}
where
\begin{equation}\label{rnqn1}
\begin{split}
R_{n}(z):=&-\frac{A_{n}'(z)}{A_{n}(z)}-v_{0}'(z)+\frac{\alpha}{z},~~~~~~~~~~~~~~~~~~~~~~~~~~~~~~~~~~~~~~~~~~~~~~~~~~~~~~~~~~\\
\end{split}
\end{equation}
\begin{equation}\label{rnqn2}
\begin{split}
Q_{n}(z):=&B_{n}'(z)-\frac{A_{n}'(z)B_{n}(z)}{A_{n}(z)}-\left[v_{0}'(z)+B_{n}(z)-\frac{\alpha}{z}\right]B_{n}(z)+\beta_{n}A_{n}(z)A_{n-1}(z).\\
\end{split}
\end{equation}
Furthermore,
{\small\begin{equation}\label{An'}
\begin{split}
A_{n}(z)=&\frac{4Ns}{h_{n}}\int_{-\infty}^{\infty}\left(z^{2}+zy+y^{2}\right)P_{n}^{2}(y)w(y)dy+\frac{2N(1-s)}{h_{n}}\int_{-\infty}^{\infty}P_{n}^{2}(y)w(y)dy\\
=&4Nsz^{2}+\frac{4Ns}{h_{n}}\int_{-\infty}^{\infty}\left[P_{n+1}(y)+\beta_{n}P_{n-1}(y)\right]\left[P_{n+1}(y)+\beta_{n}P_{n-1}(y)\right]w(y)dy\\
&+\frac{4Nsz}{h_{n}}\int_{-\infty}^{\infty}\left[P_{n+1}(y)+\beta_{n}P_{n-1}(y)\right]P_{n}(y)w(y)dy+2N(1-s)\\
=&4Nsz^{2}+4Ns\left(\beta_{n+1}+\beta_{n}\right)+2N(1-s),
\end{split}
\end{equation}}
since
\begin{equation*}
\frac{v_{0}'(z)-v_{0}'(y)}{z-y}=4Ns\left(z^{2}+zy+y^{2}\right)+2N(1-s).
\end{equation*}
Similarly, we get
\begin{equation}\label{Bn'}
B_{n}(z)=4Ns\beta_{n}z+\frac{\alpha\left[1-(-1)^{n}\right]}{2z}.
\end{equation}
Substituting Eqs. (\ref{vy}) and (\ref{An'}) into Eq. (\ref{rnqn1}), we obtain
\begin{equation*}
R_{n}(z)=-4Nsz^{3}-2N(1-s)z+\frac{\alpha}{z}-\frac{2z}{z^{2}+\frac{1-s}{2s}+\beta_{n}+\beta_{n+1}}.
\end{equation*}
Substituting Eqs. (\ref{vy}), (\ref{An'}) and (\ref{Bn'}) into Eq. (\ref{rnqn2}), we get
\begin{equation*}
\begin{split}
Q_{n}(z)=&4Ns\beta_{n}\left[1+\alpha(-1)^{n}\right]+16N^{2}s^{2}\beta_{n}\left[\frac{1-s}{2s}+\beta_{n}+\beta_{n-1}\right]\left[\frac{1-s}{2s}+\beta_{n}+\beta_{n+1}\right]\\
&-\alpha\left[1-(-1)^{n}\right]\left[N(1-s)+\frac{1}{2z^{2}}\right]-\frac{8Ns\beta_{n}z^{2}+\alpha\left[1-(-1)^{n}\right]}{z^{2}+\frac{1-s}{2s}+\beta_{n}+\beta_{n+1}}\\
&+Nsz^{2}\big\{16Ns\beta_{n}\left(\beta_{n+1}+\beta_{n}+\beta_{n-1}\right)-2\alpha\left[1-(-1)^{n}\right]+8N(1-s)\beta_{n}\big\}.
\end{split}
\end{equation*}
The proof will be completed by using Eq. (\ref{eq3.1}) to instead the last term of the above formula.
\end{proof}

Since from Theorem \ref{thm1} we have
\begin{equation*}
\beta_{n}(s)=\frac{1}{2\sqrt{3Ns}}n^{\frac{1}{2}}+\mathcal{O}(1),~~~ n\rightarrow\infty,
\end{equation*}
it follows from Eq.(\ref{newnew2}) and Eq.(\ref{newnew3}), as $n\rightarrow\infty$, that
\begin{equation*}
\begin{split}
R_{n}(z)&=-4Nsz^{3}-2N(1-s)z+\frac{\alpha}{z}+\mathcal{O}\left(n^{-\frac{1}{2}}\right),\\
Q_{n}(z)&=\left(\frac{4N^{\frac{1}{3}}s^{\frac{1}{3}}n}{3}\right)^{\frac{3}{2}}+\mathcal{O}(n).
\end{split}
\end{equation*}
Now we reconsider the Eq. (\ref{newnew1}), with denoting $P_{n}(z)$ by $\widetilde{P}_{n}(z)$ as $n\rightarrow\infty$,
\begin{equation}\label{heun1}
\widetilde{P}_{n}''(z)-\left[4Nsz^{3}+2N(1-s)z-\frac{\alpha}{z}\right]\widetilde{P}_{n}'(z)+\left[\frac{4(Ns)^{\frac{1}{3}}n}{3}\right]^{\frac{3}{2}}\widetilde{P}_{n}(z)=0.
\end{equation}

Let
\begin{equation}\label{r0}
\kappa:=(Ns)^{\frac{1}{3}}n,
\end{equation}
where $Ns\rightarrow0^{+}$, $n\rightarrow\infty$, such that $\kappa$ fixed. We find the Eq. (\ref{heun1}) is equivalent to the biconfluent Heun equation (BHE) cf. [\cite{C4'}, \S 31.12, 31.12.3]
\begin{equation}\label{heun2}
u''(x)+\left(\frac{\gamma}{x}+\delta+x\right)u'(x)+\left(\eta-\frac{\rho}{x}\right)u(x)=0
\end{equation}
through the transformation
\begin{equation*}
\widetilde{P}_{n}(z;s,\alpha)=u(x;\gamma,\delta,\eta,\rho), ~~~x=\sqrt{2}z^{2},
\end{equation*}
with parameters
\begin{equation}\label{parameters}
\gamma=-\frac{1}{2}-\frac{\alpha}{2},~~\delta=\frac{\sqrt{2}N(1-s)}{2},~~\eta=0,~~ {\rm and}~~\rho=-\frac{\sqrt{6}}{9}\kappa^{\frac{3}{2}}.
\end{equation}

\begin{rem}
 The BHE is widely encountered in contemporary mathematics and physics research. For example, many Schr\"{o}dinger equations can be solved applying the BHE, and in atomic and nuclear physics, the BHE frequently appears in studying the motion of quantum particles in $1-$, $2-$ or $3-$dimensional confinement potentials, see [\cite{Ron}]. The solutions of the BHE attract many authors, see e.g. [\cite{heun1,heun2,heun3}]. Based on [\cite{heun1}], we show that Eq. (\ref{heun2}), with parameters given in Eq. (\ref{parameters}), has the solutions in terms of the Hermite functions, reads (${\rm i}=\sqrt{-1}$)
{\small\begin{equation}\label{hermite}
u(x)=k_{1}\sum_{j=0}^{\infty}c_{j}H_{\frac{2j-\alpha-1}{2}}\left({\rm i}\left(\frac{x}{\sqrt{2}}+\frac{N(1-s)}{2}\right)\right)+
k_{2}\sum_{j=0}^{\infty}c_{j}H_{\frac{2j-\alpha-1}{2}}\left(-{\rm i}\left(\frac{x}{\sqrt{2}}+\frac{N(1-s)}{2}\right)\right)
\end{equation}}
\noindent where $k_{1}$ and $k_{2}$ are constants, whilst the coefficients $c_{n}$ satisfied a three term recurrence relation:
\begin{equation*}
L_{j}c_{j}+M_{j-1}c_{j-1}+T_{j-2}c_{j-2}=0,
\end{equation*}
with initial conditions $c_{-2}=c_{-1}=0$, and here
\begin{equation*}
\begin{split}
&L_{j}=-{\rm i}\frac{j(2j-\alpha-1)}{\sqrt{2}},\\
&M_{j}=\mp\left[\frac{9N(1-s)(2j-\alpha-1)-4\sqrt{3}\kappa^{\frac{3}{2}}}{18\sqrt{2}}\right],\\
&T_{j}=-{\rm i}\frac{2j-\alpha-1}{2\sqrt{2}},
\end{split}
\end{equation*}
where the signs $\mp$ in $M_{j}$ corresponding to $\pm{\rm i}$ in Eq. (\ref{hermite}). See details in [\cite{heun1}].
\end{rem}

\begin{rem}
Note that if we make the transformation $\widetilde{P}_{n}(z)=\widehat{P}_{n}(y)$ in Eq. (\ref{heun1}), with $y=\left(\frac{4\kappa}{3}\right)^{\frac{3}{4}}z$, then in the limit as $n\rightarrow\infty$ ($Ns\rightarrow0$, with $\kappa$ fixed), we obtain
\begin{equation*}
\widehat{P}_{n}''(y)+\frac{\alpha}{y}\widehat{P}_{n}'(y)+\widehat{P}_{n}(y)=0,
\end{equation*}
which has the solution
\begin{equation*}
\widehat{P}_{n}(y)=y^{\frac{1-\alpha}{2}}\left[c_{1}J_{\frac{1-\alpha}{2}}(y)+c_{2}Y_{\frac{1-\alpha}{2}}(y)\right],
\end{equation*}
where $J_{n}(x)$, $Y_{n}(x)$ are corresponding with the Bessel function of the first and second kind, respectively. Hence the polynomials $P_{n}(z)$ might have Mehler-Heine type asymptotic formulae.
\end{rem}

\section{Existence uniqueness of positive solutions}
In the paper [\cite{ANSA}], Alsulami \textit{et al}. discussed the existence and uniqueness of a positive solution for the nonhomogeneous nonlinear second order difference equations of the form,
\begin{equation}\label{type}
\ell_{n}=x_{n}\left(\sigma_{n,1}x_{n+1}+\sigma_{n,0}x_{n}+\sigma_{n,-1}\beta_{n-1}\right)+\kappa_{n}x_{n}
\end{equation}
with the initial conditions $x_{0}\in\mathbb{R},~ x_{1}\geq0$, whilst $\kappa_{n}\in\mathbb{R}$, $\sigma_{n,j}>0,~j\in\{0,\pm1\}$, or $\left\{\sigma_{n,0}>0, \sigma_{n,-1}\geq0,\sigma_{n,1}\geq0\right\}$. Following the results presented in [\cite{ANSA}], the solution of Eq. (\ref{eq3.1}) is obtained.

\begin{thm}
For $\alpha>-1$ and $\beta_{0}=0$, there exists a unique positive $\beta_{1}$, restricted by
\begin{equation}\label{beta1}
\beta_{1}(s;\alpha,N)=\frac{\mu_{2}(s;\alpha,N)}{\mu_{0}(s;\alpha,N)}=\frac{\mu_{0}(s;\alpha+2,N)}{\mu_{0}(s;\alpha,N)}
=\frac{\alpha+1}{2\sqrt{2Ns}}\frac{D_{-\frac{\alpha+3}{2}}\left[\frac{N(1-s)}{\sqrt{2Ns}}\right]}
{D_{-\frac{\alpha+1}{2}}\left[\frac{N(1-s)}{\sqrt{2Ns}}\right]}>0,
\end{equation}
such that the sequence $\left\{\beta_{n}(s;\alpha,N)\right\}$ satisfied by
\begin{equation}\label{eq3.111}
\beta_{n}\left(\beta_{n+1}+\beta_{n}+\beta_{n-1}+\frac{1-s}{2s}\right)=\frac{n+\alpha\Delta_{n}}{4Ns},~~~\Delta_{n}=\frac{1-(-1)^{n}}{2},
\end{equation}
is positive and the solution increases.
\end{thm}
\begin{proof}
Obviously, Eq. (\ref{eq3.111}) is a special case of Eq. (\ref{type}) with
\begin{equation*}
\sigma_{n,j}=1,~j\in\{0,\pm1\}, ~~\kappa_{n}=\frac{1-s}{2s},~~\ell_{n}=\frac{n+\alpha\Delta_{n}}{4Ns},
\end{equation*}
where $\Delta_{n}=\frac{1-(-1)^{n}}{2}$. The existence of $\beta_{1}(s;\alpha,N)>0$ such that the sequence $\left\{\beta_{n}(s;\alpha,N)\right\}$ is positive follows directly from [\cite{ANSA}, Thm. 4.1] and the uniqueness of the solution of Eq. (\ref{eq3.111}) follows immediately from [\cite{ANSA}, Thm. 5.2].
\end{proof}
\begin{figure}[!ht]
\centering
\includegraphics[width=0.65\textwidth]{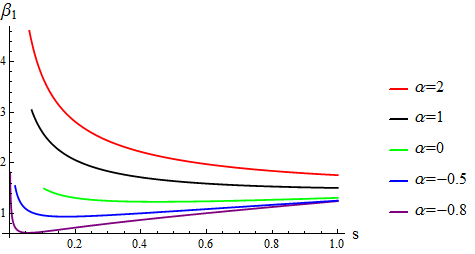}
\caption{Plots of $\beta_{1}(s;\alpha,1)$, $s\in[0,1]$ with various $\alpha$.}
\end{figure}
\begin{figure}[!ht]
\centering
\includegraphics[width=0.65\textwidth]{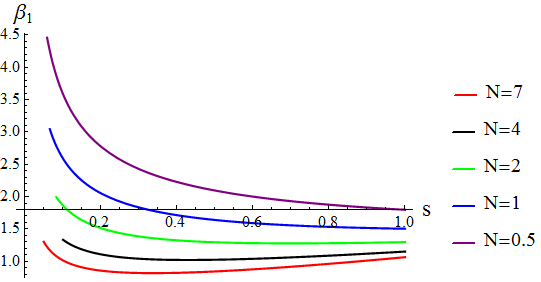}
\caption{Plots of $\beta_{1}(s;1,N)$, $s\in[0,1]$ with various $N$.}
\end{figure}

\begin{figure}[!ht]
\centering
\begin{minipage}[c]{0.33333333\textwidth}
\centering
\includegraphics[height=3.4cm,width=4.77cm]{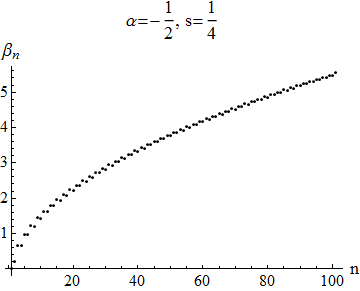}
\end{minipage}%
\begin{minipage}[c]{0.33333333\textwidth}
\centering
\includegraphics[height=3.4cm,width=4.77cm]{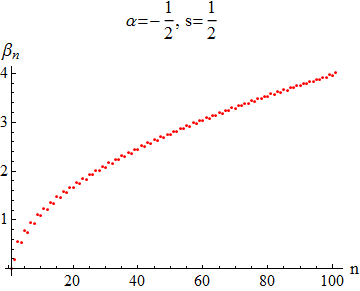}
\end{minipage}%
\begin{minipage}[c]{0.33333333\textwidth}
\centering
\includegraphics[height=3.4cm,width=4.77cm]{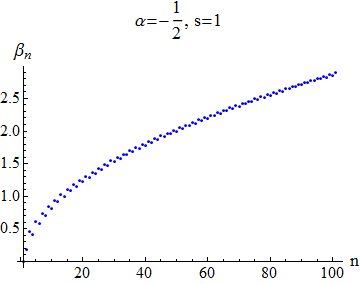}
\end{minipage}
\end{figure}
\begin{figure}[!ht]
\centering
\begin{minipage}[c]{0.33333333\textwidth}
\centering
\includegraphics[height=3.4cm,width=4.77cm]{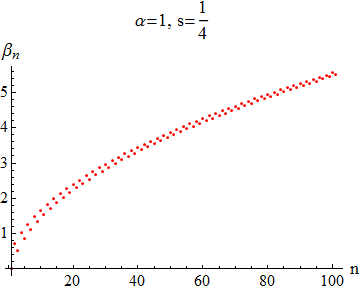}
\end{minipage}%
\begin{minipage}[c]{0.33333333\textwidth}
\centering
\includegraphics[height=3.4cm,width=4.77cm]{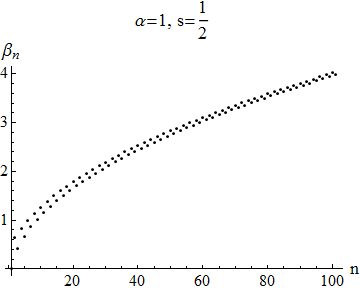}
\end{minipage}%
\begin{minipage}[c]{0.33333333\textwidth}
\centering
\includegraphics[height=3.4cm,width=4.77cm]{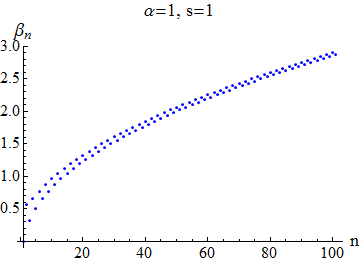}
\end{minipage}
\end{figure}
\begin{figure}[!ht]
\centering
\begin{minipage}[c]{0.33333333\textwidth}
\centering
\includegraphics[height=3.4cm,width=4.77cm]{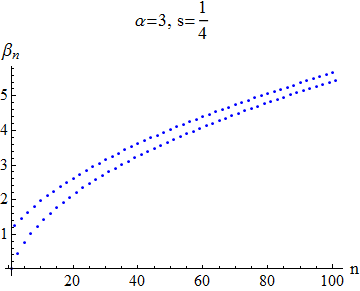}
\end{minipage}%
\begin{minipage}[c]{0.33333333\textwidth}
\centering
\includegraphics[height=3.4cm,width=4.77cm]{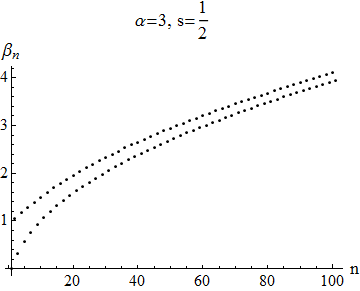}
\end{minipage}%
\begin{minipage}[c]{0.33333333\textwidth}
\centering
\includegraphics[height=3.4cm,width=4.77cm]{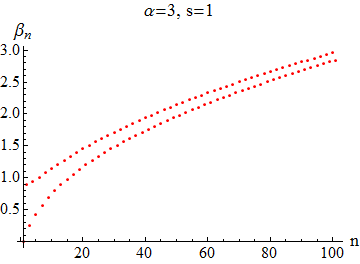}
\end{minipage}
\caption{Plots of the recurrence coefficients $\beta_{n}(s;\alpha,1)$, in given initial condition (\ref{beta1}) with various choices of $(s;\alpha)$, and $N=1$.}
\end{figure}
\begin{figure}[!ht]
\centering
\begin{minipage}[c]{0.33333333\textwidth}
\centering
\includegraphics[height=3.4cm,width=4.77cm]{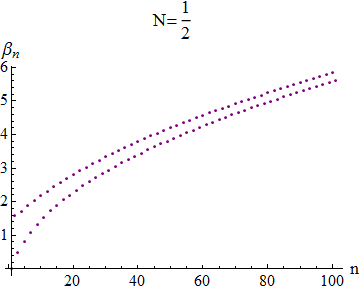}
\end{minipage}%
\begin{minipage}[c]{0.33333333\textwidth}
\centering
\includegraphics[height=3.4cm,width=4.77cm]{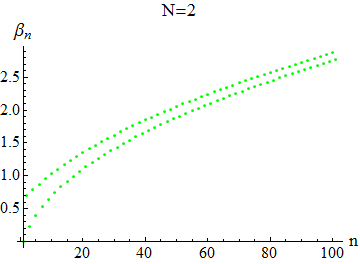}
\end{minipage}%
\begin{minipage}[c]{0.33333333\textwidth}
\centering
\includegraphics[height=3.4cm,width=4.77cm]{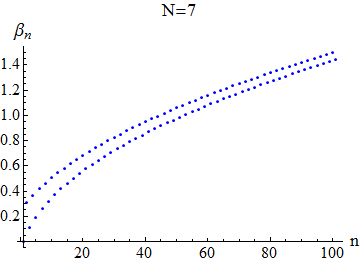}
\end{minipage}
\caption{Plots of the recurrence coefficients $\beta_{n}\left(\frac{1}{2};3,N\right)$, in given initial condition (\ref{beta1}) with various choices of $N$, and $s=\frac{1}{2},\alpha=3$.}
\end{figure}

\begin{figure}[!ht]
\centering
\begin{minipage}[c]{0.5\textwidth}
\centering
\includegraphics[height=4.5cm,width=6.7cm]{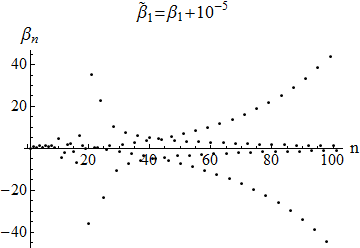}
\end{minipage}%
\begin{minipage}[c]{0.5\textwidth}
\centering
\includegraphics[height=4.5cm,width=6.7cm]{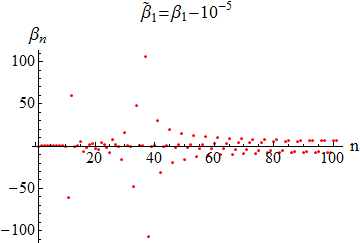}
\end{minipage}
\end{figure}
\begin{figure}[!ht]
\centering
\begin{minipage}[c]{0.5\textwidth}
\centering
\includegraphics[height=4.5cm,width=6.7cm]{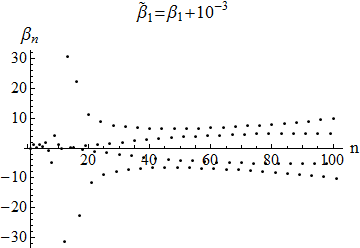}
\end{minipage}%
\begin{minipage}[c]{0.5\textwidth}
\centering
\includegraphics[height=4.5cm,width=6.7cm]{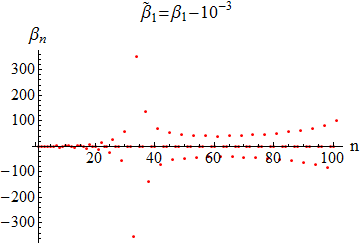}
\end{minipage}
\end{figure}
\begin{figure}[!ht]
\centering
\begin{minipage}[c]{0.5\textwidth}
\centering
\includegraphics[height=4.5cm,width=6.7cm]{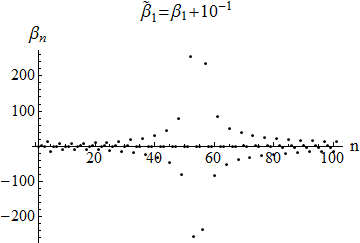}
\end{minipage}%
\begin{minipage}[c]{0.5\textwidth}
\centering
\includegraphics[height=4.5cm,width=6.7cm]{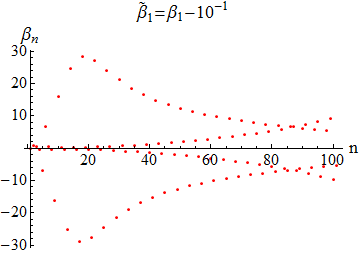}
\end{minipage}
\caption{Plots of the recurrence coefficients $\beta_{n}\left(\frac{1}{2};3,1\right)$, with initial conditions $\widetilde{\beta}_{0}\left(\frac{1}{2};3,1\right)=0,\widetilde{\beta_{1}}\left(\frac{1}{2};3,1\right)=\beta_{1}\left(\frac{1}{2};3,1\right)
\pm\varepsilon,\varepsilon\in\left\{10^{-1},10^{-3},10^{-5}\right\}$.}
\end{figure}

\begin{figure}[!ht]
\centering
\begin{minipage}[c]{0.5\textwidth}
\centering
\includegraphics[height=4.5cm,width=6.7cm]{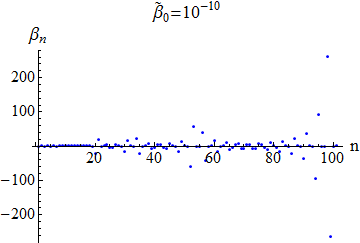}
\end{minipage}%
\begin{minipage}[c]{0.5\textwidth}
\centering
\includegraphics[height=4.5cm,width=6.7cm]{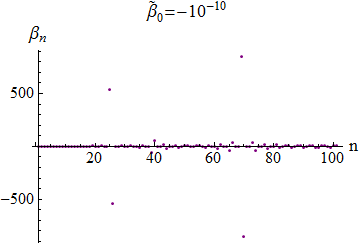}
\end{minipage}
\caption{Plots of the recurrence coefficients $\beta_{n}\left(\frac{1}{2};3,1\right)$, with initial conditions $\widetilde{\beta}_{0}\left(\frac{1}{2};3,1\right)=\pm10^{-10},\widetilde{\beta_{1}}\left(\frac{1}{2};3,1\right)=\beta_{1}\left(\frac{1}{2};3,1\right)$.}
\end{figure}

\newpage
\begin{rem}
The figures are plotted by Mathematica taking 500 digits precision.
\end{rem}

\newpage

\section{The Hankel determinant}
The well-known supplementary conditions ($S_{1}$), ($S_{2}$) and the ``sum rule'' ($S'_{2}$) for the ladder operators, satisfied by the functions $A_{n}(z)$ and $B_{n}(z)$, continue to hold by a slight modification. These are,
\begin{equation*}
~~~~~~~~~~~~~~~~~~~~~~~~B_{n+1}(z)+B_{n}(z)=zA_{n}(z)-v_{0}'(z)+\frac{\alpha}{z}, ~~~~~~~~~~~~~~~~~~~~~~~~~~~~~~~~~~~~~~~~~~~(S_{1})
\end{equation*}
\begin{equation*}
~~~~~~~~~~~~~~~~1+z\left[B_{n+1}(z)-B_{n}(z)\right]=\beta_{n+1}(s)A_{n+1}(z)-\beta_{n}(s)A_{n-1}(z), ~~~~~~~~~~~~~~~~~~~~~~~~~~(S_{2})
\end{equation*}
and
\begin{equation*}
~~~~~~~~~~~~~~B_{n}^{2}(z)+\left[v_{0}'(z)-\frac{\alpha}{z}\right]B_{n}(z)+\sum_{j=0}^{n-1}A_{j}(z)=\beta_{n}(s)A_{n}(z)A_{n-1}(z). ~~~~~~~~~~~~~~~~~~~~~~~(S'_{2})
\end{equation*}
where $v_{0}(z)=Nsz^{4}+N(1-s)z^{2},~z\in\mathbb{R},~s\in[0,1]$. $A_{n}(z)$ and $B_{n}(z)$ have been given in Eq. (\ref{An'}) and Eq. (\ref{Bn'}), which read
\begin{equation}\label{An'''}
A_{n}(z)=4Nsz^{2}+4Ns\left(\beta_{n+1}+\beta_{n}\right)+2N(1-s),
\end{equation}
and
\begin{equation}\label{Bn'''}
B_{n}(z)=4Ns\beta_{n}z+\frac{\alpha}{z}\Delta_{n},~{\rm with}~\Delta_{n}=\frac{\left[1-(-1)^{n}\right]}{2}.
\end{equation}
The three identities (S$_{1}$), (S$_{2}$) and (S$'_{2}$) were also stated in Magnus [\cite{Magnus2}]. Based on this and the formula (\ref{Dn}), we first display the relation between $D_{n}(s;\alpha,N)$ and $\beta_{n}(s;\alpha,N)$ in the next lemma.
\begin{lem}\label{lemhankel1}
The logarithmic derivative of the Hankel determinant can be expressed as
\begin{equation}\label{determinant001}
\begin{split}
\frac{d}{ds}\log D_{n}(s)=&\frac{N\left(s^{2}-1\right)(n+\alpha\Delta_{n})-2sn(n+\alpha)}{8s^{2}}+\frac{(1+s)\left(n+\alpha\Delta_{n}\right)^{2}}{32s^{2}\beta_{n}}\\
&+\frac{N^{2}\left(1-s^{2}\right)^{2}+2Ns(1+s)^{2}\left(n+2\alpha-3\alpha\Delta_{n}\right)-4s^{2}}{8s^{2}(1+s)}\beta_{n}\\
&+\frac{N^{2}(1-s^{2})}{2s}\beta_{n}^{2}+\frac{1}{2}N^{2}(1+s)\beta_{n}^{3}-\frac{2s}{1+s}\left(\frac{s{\beta_{n}'}^{2}}{\beta_{n}}+\beta_{n}'\right).
\end{split}
\end{equation}
\end{lem}
\begin{proof}
It follows from Eq. (\ref{Dn}) that,
\begin{equation}\label{determinant002}
\log D_{n}(s)=\sum_{j=0}^{n-1}\log h_{j}(s),
\end{equation}
we find by taking the derivative of Eq. (\ref{determinant002}) with respect to $s$,
\begin{equation*}
\begin{split}
&\frac{D_{n}'(s)}{D_{n}(s)}=\sum_{j=0}^{n-1}\frac{h'_{j}(s)}{h_{j}(s)}=\frac{N(1+s)}{2s}\sum_{j=0}^{n-1}\left(\beta_{j+1}+\beta_{j}\right)-\frac{n^{2}}{4s}-\frac{n\alpha}{4s}.\\
\end{split}
\end{equation*}
Here we have used the identity (\ref{pro3.2eq1}), which gives
\begin{equation*}
\frac{h'_{n}(s)}{h_{n}(s)}=\frac{N(1+s)}{2s}\left(\beta_{n+1}+\beta_{n}\right)-\frac{2n+1}{4s}-\frac{\alpha}{4s}.
\end{equation*}
Substituting Eqs. (\ref{An'''}), (\ref{Bn'''}) and $v_{0}(x)=Nsx^{4}+N(1-s)x^{2}$ into ($S'_{2}$), we obtain
{\small\begin{equation*}
\begin{split}
&\left[4nNs+16N^{2}s^{2}\beta_{n}^{2}+4Ns\alpha\Delta_{n}+8N^{2}s(1-s)\beta_{n}\right]z^{2}+4Ns\alpha\beta_{n}(2\Delta_{n}-1)+2N(1-s)\alpha\Delta_{n}\\
&+2nN(1-s)+4Ns\sum_{j=0}^{n-1}\left(\beta_{j+1}+\beta_{j}\right)=16N^{2}s^{2}\beta_{n}\left[\beta_{n+1}+2\beta_{n}+\beta_{n-1}+\frac{1-s}{s}\right]z^{2}\\
&+16N^{2}s^{2}\beta_{n}\left[\beta_{n+1}+\beta_{n}+\frac{1-s}{2s}\right]\left[\beta_{n}+\beta_{n-1}+\frac{1-s}{2s}\right].
\end{split}
\end{equation*}}
Equating coefficients on $z$ gives
\begin{equation}\label{residues2}
\begin{split}
\beta_{n}\left(\beta_{n+1}+\beta_{n}+\beta_{n-1}+\frac{1-s}{2s}\right)=\frac{n+\alpha\Delta_{n}}{4Ns},
\end{split}
\end{equation}
\begin{equation}\label{residues1}
\begin{split}
\sum_{j=0}^{n-1}\left(\beta_{j+1}+\beta_{j}\right)=&4Ns\beta_{n}\left[\beta_{n+1}+\beta_{n}+\frac{1-s}{2s}\right]\left[\beta_{n}+\beta_{n-1}+\frac{1-s}{2s}\right]\\
&+\alpha\beta_{n}(1-2\Delta_{n})-\frac{1-s}{2s}\left(n+\alpha\Delta_{n}\right).
\end{split}
\end{equation}
Hence we have
\begin{equation}\label{determinant003}
\begin{split}
\frac{d}{ds}\log D_{n}(s)=&2N^{2}(1+s)\beta_{n}\left(\beta_{n+1}+\beta_{n}+\frac{1-s}{2s}\right)\left(\beta_{n}+\beta_{n-1}+\frac{1-s}{2s}\right)\\
&+\frac{N(1+s)}{2s}\alpha\beta_{n}(1-2\Delta_{n})-\frac{N\left(1-s^{2}\right)}{4s^{2}}(n+\alpha\Delta_{n})-\frac{n(n+\alpha)}{4s}.
\end{split}
\end{equation}
From the Eq. (\ref{residues2}), we find
\begin{equation*}
\begin{split}
\beta_{n}\left(\beta_{n+1}+\beta_{n}+\frac{1-s}{2s}\right)=\frac{n+\alpha\Delta_{n}}{4Ns}-\beta_{n}\beta_{n-1}.
\end{split}
\end{equation*}
Substituting this into Eq. (\ref{determinant003}) and eliminating $\beta_{n-1}$ by Eq. (\ref{betan6}), the result is obtained.
\end{proof}

Let $N\rightarrow\infty,~n\rightarrow\infty$ but the term $n/N$  tends to a fixed number $r$, which is bounded away from $0$ and close to $1$. If we substitute the asymptotic expansion of $\beta_{n}(s)$, given by Theorem \ref{nN}, into the right hand of Eq. (\ref{determinant001}), we find
\begin{equation}\label{2018-1-new}
\begin{split}
\frac{D_{n}'(s)}{D_{n}(s)}=A(s)N^{2}+B(s)N+C(s)+\frac{D(s)}{N}+\mathcal{O}\left(N^{-2}\right),~N\rightarrow\infty,
\end{split}
\end{equation}
where the constants $A(s),B(s),C(s)$ and $D(s)$ can be found in Appendix B.

Therefore integrating from $s=0$ to $s=1$, and substituting the constants given in Appendix B, we have
\begin{lem}\label{newlemma1}
For any $\alpha>-1$, as $n\rightarrow\infty$ with $n=rN$, $r$ fixed, we have
\begin{equation}\label{2018-1}
\begin{split}
&\log\frac{D_{n}(1)}{D_{n}(0)}=\int_{0}^{1}\frac{d}{ds}\log D_{n}(s)ds\\
&=\frac{\left(3-2\log3-2\log r\right)}{8}n^{2}+\frac{\alpha\left(-1+\log3+\log r\right)}{4}n+\frac{\left(3\alpha^{2}-1\right)\log2}{12}-\frac{\alpha^{2}\log3}{4}\\
&~~~+\frac{\alpha\left(\alpha^{2}+3\right)}{48}\frac{1}{n}+\mathcal{O}\left(n^{-2}\right),~~~\text{$n$ {\rm even}},\\
\end{split}
\end{equation}
and
\begin{equation}\label{2018-2}
\begin{split}
\log\frac{D_{n}(1)}{D_{n}(0)}=&\frac{\left(3-2\log3-2\log r\right)}{8}n^{2}-\frac{\alpha\left(-1+\log3+\log r\right)}{4}n+\frac{\left(3\alpha^{2}-1\right)\log2}{12}\\
&-\frac{\alpha^{2}\log3}{4}
+\frac{\alpha\left(\alpha^{2}-3\right)}{48}\frac{1}{n}+\mathcal{O}\left(n^{-2}\right),~~~\text{$n$ {\rm odd}}.\\
\end{split}
\end{equation}
\end{lem}

\begin{proof}
We calculate the integrals in Eq. (\ref{2018-1-new}) with the aid of Mathematica,
\begin{equation*}
\begin{split}
&\log\frac{D_{n}(1)}{D_{n}(0)}=\int_{0}^{1}\frac{d}{ds}\log D_{n}(s)ds\\
&=N^{2}\int_{0}^{1}A(s)ds+N\int_{0}^{1}B(s)ds+\int_{0}^{1}C(s)ds+\frac{1}{N}\int_{0}^{1}D(s)ds+\mathcal{O}\left(N^{-2}\right)\\
&=\frac{r^{2}\left(3-2\log3-2\log r\right)}{8}N^{2}+\frac{r\alpha\left(-1+\log3+\log r\right)}{4}N+\frac{\left(3\alpha^{2}-1\right)\log2}{12}-\frac{\alpha^{2}\log3}{4}\\
&~~~+\frac{\alpha\left(\alpha^{2}+3\right)}{48r}\frac{1}{N}+\mathcal{O}\left(N^{-2}\right)\\
&=\frac{\left(3-2\log3-2\log r\right)}{8}n^{2}+\frac{\alpha\left(-1+\log3+\log r\right)}{4}n+\frac{\left(3\alpha^{2}-1\right)\log2}{12}-\frac{\alpha^{2}\log3}{4}\\
&~~~+\frac{\alpha\left(\alpha^{2}+3\right)}{48}\frac{1}{n}+\mathcal{O}\left(n^{-2}\right),~~~\text{$n$ even},\\
\end{split}
\end{equation*}
Similarly, we get the result for the case $n$ is odd.
\end{proof}

Mehta and Normand [\cite{MehtaNormand}] gives
\begin{equation*}
\begin{split}
D_{n}(0)&=\frac{1}{n!}\int_{\mathbb{R}^{n}}\prod_{1\leq j<k\leq n}\left(x_{k}-x_{j}\right)^{2}\prod_{\ell=1}\left|x_{\ell}\right|^{\alpha}{\rm e}^{-Nx_{\ell}^{2}}dx_{\ell}\\
&=\frac{(2\pi)^{\frac{n}{2}}}{(2N)^{\frac{n^{2}}{2}}N^{\frac{\alpha n}{2}}\Gamma(n+1)}\prod_{j=1}^{n}\frac{\Gamma\left(\frac{\alpha+1}{2}+\left\lfloor\frac{j}{2}\right\rfloor\right)}{\Gamma\left(\frac{1}{2}+\left\lfloor\frac{j}{2}\right\rfloor\right)}j!,\\
\end{split}
\end{equation*}
where $\lfloor n\rfloor$ denotes the integer part of $n$.

Alternatively, when $n$ is even, Han and Chen [\cite{HanChen}], see also [\cite{DeanoSimm}], gives the expression of $D_{n}(0)$ as
\begin{equation}\label{2018-3}
\begin{split}
D_{n}(0)=\frac{(2\pi)^{\frac{n}{2}}}{(2N)^{\frac{n^{2}}{2}}N^{\frac{\alpha n}{2}}}\frac{G\left(\frac{3}{2}\right)G\left(\frac{1}{2}\right)}{G\left(\frac{\alpha+3}{2}\right)G\left(\frac{\alpha+1}{2}\right)}\frac{G(n+1)G\left(\frac{\alpha+n+3}{2}\right)
G\left(\frac{\alpha+n+1}{2}\right)}{G\left(\frac{n+3}{2}\right)G\left(\frac{n+1}{2}\right)},
\end{split}
\end{equation}
where $G$ is the Barnes $G-$function, that satisfies the functional relation $G(z+1)=\Gamma(z)G(z)$, with $G(0)=1$, see Voros [\cite{Voros}]. It is well-known that the Barnes G-function has the asymptotic expansion:
\begin{equation}\label{2018-5}
\begin{split}
\log G(z+1)\simeq& \frac{z^{2}}{4}+z\log \Gamma(z+1)-\left[\frac{z(z+1)}{2}+\frac{1}{12}\right]\log z-\log A\\
&+\sum_{k=1}^{\infty}\frac{\mathbb{B}_{2k+2}}{2k(2k+1)(2k+2)z^{2k}},
\end{split}
\end{equation}
as $z\rightarrow\infty$ with $|\arg z|<\pi$. Here $\mathbb{B}_{n}$ are Bernoulli numbers and
\begin{equation*}
\begin{split}
A=\exp\left(\frac{1}{12}-\zeta'(-1)\right),
\end{split}
\end{equation*}
is the Glaisher-Kinkelin constant, $A=1.2824271291\ldots$, where $\zeta(x)$ is the Riemann zeta function, see for example [\cite{C4'}, E.q. 5.17.5].
Then substituting Eq. (\ref{2018-5}) into Eq. (\ref{2018-3}), we obtain,
\begin{equation}\label{2018-7}
\begin{split}
\log D_{n}(0)=&\frac{2\log r-2\log2-3}{4}n^{2}+\left[\log(2\pi)-\frac{\alpha\left(1+\log2-\log r\right)}{2}\right]n\\
&+\frac{3\alpha^{2}-1}{12}\log n-\log A+\frac{1+6\alpha\log(2\pi)-3\alpha^{2}\log2}{12}\\
&+\log\frac{G\left(\frac{3}{2}\right)G\left(\frac{1}{2}\right)}{G\left(\frac{\alpha+3}{2}\right)G\left(\frac{\alpha+1}{2}\right)}+\frac{\alpha^{3}+\alpha}{12n}+\mathcal{O}\left(n^{-2}\right),~\text{$n$ even}.
\end{split}
\end{equation}
A similar calculation gives
\begin{equation*}
\begin{split}
D_{n}(0)=\frac{(2\pi)^{\frac{n}{2}}}{(2N)^{\frac{n^{2}}{2}}N^{\frac{\alpha n}{2}}}\frac{G\left(\frac{3}{2}\right)G\left(\frac{1}{2}\right)}{G\left(\frac{\alpha+3}{2}\right)G\left(\frac{\alpha+1}{2}\right)}\frac{G(n+1)\left[G\left(\frac{\alpha+n+2}{2}\right)\right]^{2}
}{\left[G\left(\frac{n+2}{2}\right)\right]^{2}},~~~\text{$n$ odd}.
\end{split}
\end{equation*}
which has the asymptotic expansion
\begin{equation}\label{2018-8}
\begin{split}
\log D_{n}(0)=&\frac{2\log r-2\log2-3}{4}n^{2}+\left[\log(2\pi)-\frac{\alpha\left(1+\log2-\log r\right)}{2}\right]n\\
&+\frac{3\alpha^{2}-1}{12}\log n-\log A+\frac{1+6\alpha\log(2\pi)-3\alpha^{2}\log2}{12}\\
&+\log\frac{G\left(\frac{3}{2}\right)G\left(\frac{1}{2}\right)}{G\left(\frac{\alpha+3}{2}\right)G\left(\frac{\alpha+1}{2}\right)}+\frac{\alpha^{3}-2\alpha}{12n}+\mathcal{O}\left(n^{-2}\right),~\text{$n$ odd}.
\end{split}
\end{equation}
Consequently, we get the asymptotics expansion for the Hankel determinant $D_{n}(1)$,
\begin{thm}As $r=\frac{n}{N}$ tends to $1$, we have
\begin{equation}\label{2018-9}
\begin{split}
\log D_{n}(1)=&\frac{2\log r-\log144-3}{8}n^{2}+\frac{1}{4}\left[3\alpha\log r+4\log(2\pi)-\alpha\left(3+\log\frac{4}{3}\right)\right]n\\
&+\frac{3\alpha^{2}-1}{12}\log n+\frac{1}{12}\left[1-\log2+6\alpha\log(2\pi)-3\alpha^{2}\log3\right]-\log A\\
&+\log\frac{G\left(\frac{3}{2}\right)G\left(\frac{1}{2}\right)}{G\left(\frac{\alpha+3}{2}\right)G\left(\frac{\alpha+1}{2}\right)}+\frac{\alpha\left(5\alpha^{2}+7\right)}{48}
+\mathcal{O}\left(n^{-2}\right), ~~~~~~\text{$n$  {\rm even}},
\end{split}
\end{equation}
whilst
\begin{equation}\label{2018-10}
\begin{split}
\log D_{n}(1)=&\frac{2\log r-\log144-3}{8}n^{2}+\frac{1}{4}\left[\alpha\log r+4\log(2\pi)-\alpha-\alpha\log12\right]n\\
&+\frac{3\alpha^{2}-1}{12}\log n+\frac{1}{12}\left[1-\log2+6\alpha\log(2\pi)-3\alpha^{2}\log3\right]-\log A\\
&+\log\frac{G\left(\frac{3}{2}\right)G\left(\frac{1}{2}\right)}{G\left(\frac{\alpha+3}{2}\right)G\left(\frac{\alpha+1}{2}\right)}+\frac{\alpha\left(5\alpha^{2}-11\right)}{48}
+\mathcal{O}\left(n^{-2}\right), ~~~~~~\text{$n$  {\rm odd}}.
\end{split}
\end{equation}
\end{thm}
\begin{proof}
The result is obtained by direct computation based on Eqs.(\ref{2018-7}) and (\ref{2018-8}) and Lemma \ref{newlemma1}.
\end{proof}

\section{Acknowledgements}
The financial support of the Macau Science and Technology Development Fund under grant
number FDCT 130/2014/A3 and FDCT 023/2017/A1 are gratefully acknowledged. We would also like
to thank the University of Macau for generous support: MYRG 2014-00011 FST, MYRG
2014-00004 FST.

\section{Appendix A}
\begin{flalign*}
{\rm dP_{I}}~~~~~~~x_{n+1}+x_{n}+x_{n-1}=\frac{z_{n}+\gamma(-1)^{n}}{x_{n}}+\delta.
\end{flalign*}

\section{Appendix B}
Here we list the coefficients of the Eq. (\ref{2018-1-new}).

For $A(s)$, letting $\widetilde{g}(s):=s-1+g(s)$, where $g(s)=\sqrt{1-2s+12rs+s^{2}}$, then we have
\begin{equation*}
\begin{split}
A(s)=&\frac{3r^{2}(1+s)}{8s\widetilde{g}(s)}-\frac{r\left(1+2rs-s^{2}\right)}{8s^{2}}+\frac{(1+s)\left(1-2s+2rs+s^{2}\right)\widetilde{g}(s)}{96s^{3}}~~~~~~~~~~~~~~\\
&+\frac{\left(1-s^{2}\right)\widetilde{g}^{2}(s)}{288s^{3}}+\frac{(1+s)\widetilde{g}^{3}(s)}{3456s^{3}},~~~\text{$n$ even or odd}.
\end{split}
\end{equation*}

For $B(s)$, defining
\begin{equation*}
\begin{split}
B_{0}(s):=&\frac{\alpha}{6s^{2}g(s)\widetilde{g}^{2}(s)}\Big\{1-g(s)+s(6r-1)(2g(s)-3)+s^{2}\left[2+\left(6r-18r^{2}\right)(g(s)-4)\right]\\
&+s^{3}(6r-2)(g(s)-1)+s^{4}(g(s)+12r-3)+s^{5}\Big\},
\end{split}
\end{equation*}
then we obtain $B(s)=-B_{0}(s)$ if $n$ is even and $B(s)=B_{0}(s)$ if $n$ is odd.

For $C(s)$, we first define
\begin{equation*}
\begin{split}
c_{0}(s):=&\frac{1+s}{6sg\widetilde{g}^{2}\left[1+2g-2s(1-6r+g)+s^{2}\right]},~~~~~~~~~~~~~~~~~~~~~~~~~~~~~~~~~~~~~~~~~~~~~~~~~~~~~~
\end{split}
\end{equation*}
\begin{equation*}
\begin{split}
c_{1}(s):=&-1 + 6 s - 3 r s - 15 s^2 + 12 r s^2 + 72 r^2 s^2 + 20 s^3 -
 18 r s^3 - 144 r^2 s^3 - 432 r^3 s^3 - 15 s^4 \\
 &+ 12 r s^4 +
 72 r^2 s^4 + 6 s^5 - 3 r s^5 - s^6,
\end{split}
\end{equation*}
\begin{equation*}
\begin{split}
c_{2}(s):=&g\left(1 - 5 s - 3 r s + 10 s^2 + 9 r s^2 - 36 r^2 s^2 - 10 s^3 -
   9 r s^3 + 36 r^2 s^3 + 5 s^4 + 3 r s^4 - s^5\right),
\end{split}
\end{equation*}
\begin{equation*}
\begin{split}
c_{3}(s):=&g\big(-3 + 15 s - 27 r s - 30 s^2 + 81 r s^2 + 108 r^2 s^2 +
   30 s^3 - 81 r s^3 - 108 r^2 s^3 - 15 s^4\\
   & + 27 r s^4 +
   3 s^5\big)\alpha^2,
\end{split}
\end{equation*}
\begin{equation*}
\begin{split}
c_{4}(s):=&\big(3 - 18 s + 45 r s + 45 s^2 - 180 r s^2 - 60 s^3 + 270 r s^3 -
   1296 r^3 s^3 + 45 s^4 - 180 r s^4 - 18 s^5 \\
   &+ 45 r s^5 +
   3 s^6\big) \alpha^2,
\end{split}
\end{equation*}
then we find that $C(s)=c_{0}(s)\left[c_{1}(s)+c_{2}(s)+c_{3}(s)+c_{4}(s)\right]$, either $n$ is even or odd.

For $D(s)$, we get $D(s):=d_{0}(s)\left[d_{1}(s)+d_{2}(s)\right]$ when $n$ is even, with
\begin{equation*}
\begin{split}
d_{0}(s):=&-\frac{(1+s)\alpha}{g\widetilde{g}^{2}\left[1+2g-2s(1-6r+g)+s^{2}\right]^{3}},~~~~~~~~~~~~~~~~~~~~~~~~~~~~~~~~~~~~~~~~~~~~~~~~~~~~~~
\end{split}
\end{equation*}
\begin{equation*}
\begin{split}
d_{1}(s):=&g\big(-3 + 15 s - 15 r s - 30 s^2 + 45 r s^2 - 36 r^2 s^2 + 30 s^3 -
 45 r s^3 + 36 r^2 s^3 - 15 s^4\\
 &+ 15 r s^4 + 3 s^5 - \alpha^2 +
 5 s  \alpha^2 - 15 r s  \alpha^2 - 10 s^2  \alpha^2+
 45 r s^2  \alpha^2 + 36 r^2 s^2  \alpha^2 + 10 s^3  \alpha^2 \\
 &- 45 r s^3  \alpha^2
 - 36 r^2 s^3  \alpha^2 - 5 s^4  \alpha^2 +
 15 r s^4  \alpha^2 + s^5  \alpha^2\big),
\end{split}
\end{equation*}
\begin{equation*}
\begin{split}
d_{2}(s):=&3 - 18 s + 33 r s + 45 s^2 - 132 r s^2 + 72 r^2 s^2 - 60 s^3 +
 198 r s^3 - 144 r^2 s^3 - 2160 r^3 s^3\\
 & + 45 s^4 - 132 r s^4 +
 72 r^2 s^4 - 18 s^5 + 33 r s^5 + 3 s^6 + \alpha^2 -
 6 s \alpha^2 + 21 r s \alpha^2 + 15 s^2 \alpha^2 \\
 &- 84 r s^2 \alpha^2 + 36 r^2 s^2\alpha^2- 20 s^3 \alpha^2 +
 126 r s^3 \alpha^2 - 72 r^2 s^3\alpha^2 -
 864 r^3 s^3\alpha^2+ 15 s^4 \alpha^2 \\
 &- 84 r s^4 \alpha^2 +
 36 r^2 s^4\alpha^2 - 6 s^5 \alpha^2 + 21 r s^5 \alpha^2 +
 s^6 \alpha^2.
\end{split}
\end{equation*}
else if $n$ is odd, we have $D(s):=\widehat{d}_{0}(s)\left[\widehat{d}_{1}(s)+\widehat{d}_{2}(s)\right]$, where
\begin{equation*}
\begin{split}
\widehat{d}_{0}(s):=\frac{\alpha}{4f^{2}g^{5}}, ~{\rm with}~f=1-2s-4rs+s^{2},~~~~~~~~~~~~~~~~~~~~~~~~~~~~~~~~~~~~~~~~~~~~~~~~
\end{split}
\end{equation*}
\begin{equation*}
\begin{split}
\widehat{d}_{1}(s):=&g\big(1 - 4 s + 56 r s + 5 s^2 - 112 r s^2 + 272 r^2 s^2 - 5 s^4 +
 112 r s^4 - 272 r^2 s^4 + 4 s^5 \\
 &- 56 r s^5 - s^6 -\alpha^2+
 4 s \alpha^2 - 24 r s \alpha^2 - 5 s^2 \alpha^2 +
 48 r s^2 \alpha^2 - 144 r^2 s^2 \alpha^2+ 5 s^4 \alpha^2\\
 &-48 r s^4\alpha^2+ 144 r^2 s^4\alpha^2 - 4 s^5\alpha^2+
 24 r s^5 \alpha^2 + s^6 \alpha^2\big),
\end{split}
\end{equation*}
\begin{equation*}
\begin{split}
\widehat{d}_{2}(s):=&-2 + 10 s - 64 r s - 18 s^2 + 192 r s^2 - 544 r^2 s^2 + 10 s^3 -
 128 r s^3 + 544 r^2 s^3 - 768 r^3 s^3\\
 & + 10 s^4 - 128 r s^4 +
 544 r^2 s^4 - 768 r^3 s^4
 - 18 s^5 + 192 r s^5 - 544 r^2 s^5 +
 10 s^6 - 64 r s^6\\
 &- 2 s^7+  \alpha^2  - 5 s  \alpha^2  +
 32 r s  \alpha^2 + 9 s^2  \alpha^2  - 96 r s^2  \alpha^2 +
 272 r^2 s^2  \alpha^2  - 5 s^3  \alpha^2  + 64 r s^3  \alpha^2  \\
 &- 272 r^2 s^3  \alpha^2 + 384 r^3 s^3  \alpha^2 - 5 s^4  \alpha^2  +
 64 r s^4  \alpha^2 - 272 r^2 s^4  \alpha^2  +
 384 r^3 s^4  \alpha^2  + 9 s^5  \alpha^2 \\
 & - 96 r s^5  \alpha^2  +
 272 r^2 s^5  \alpha^2  - 5 s^6  \alpha^2  + 32 r s^6  \alpha^2  +
 s^7  \alpha^2 .
\end{split}
\end{equation*}

\end{document}